\documentclass[11pt]{article}

\usepackage[english]{babel}
\usepackage[margin=1in]{geometry}

\usepackage{natbib}
\usepackage{graphicx}
\usepackage{amsfonts}
\usepackage{amssymb}
\usepackage{amsmath}
\usepackage[colorlinks=true, allcolors=red]{hyperref}
\usepackage{mathtools} 
\usepackage{mathrsfs}
\usepackage{booktabs}
\usepackage{multirow}
\usepackage{float} 
\usepackage{enumitem}
\usepackage{caption}
\usepackage{subcaption}
\usepackage{algorithm}
\usepackage{algpseudocode}
\usepackage{amsthm}
\usepackage{titlesec}
\usepackage{orcidlink}
\usepackage{stmaryrd}

\newtheorem{theorem}{Theorem}
\newtheorem{proposition}{Proposition}
\newtheorem{example}{Example}
\newtheorem{illustration}{Illustration}
\newtheorem{lemma}{Lemma}
\newtheorem{remark}{Remark}
\newtheorem{definition}{Definition}
\newtheorem{assumption}{Assumption}
\newtheorem{corollary}{Corollary}
\newcommand{\indep}{\perp\!\!\!\perp}
\newcommand{\symdif}{\,\triangle\,}

\usepackage{newtxtext}
\usepackage{newtxmath}

\title{A Post-Processing Conformal Prediction Approach for Conditional Coverage via Pivotal Scores}

\author{
\Large F\'elix Laplante\textsuperscript{\large \orcidlink{0009-0001-7839-1502}} \\[0.5em]
\small Universit\'e Paris-Saclay, CNRS, Univ \'Evry, \\
\small Laboratoire de Math\'ematiques et Mod\'elisation d'\'Evry, \\
\small 91037, \'Evry-Courcouronnes, France. \\
\small \href{mailto:felixlaplante.research@gmail.com}{felixlaplante.research@gmail.com}
}

\begin{document}
\maketitle

\begin{abstract}
While Conformal Prediction (CP) has proven to be a powerful framework for uncertainty quantification, guaranteeing conditional coverage remains a central challenge. Although finite-sample, distribution-free conditional validity is known to be impossible without structural assumptions, we show that for i.i.d. data, it is fundamentally equivalent to constructing a nonconformity score whose distribution is independent of the features. This theoretical characterization motivates PIT-CP, a new post-processing correction that maps any base nonconformity score to an approximately invariant one while preserving its geometry, interpretability, and marginal coverage. This perspective is particularly appealing in practice, since it may be neither economical nor time-effective to retrain a full generative model when a strong prediction-driven model already provides highly accurate point estimates. Our procedure reduces the problem to one-dimensional conditional density estimation on the induced score, rather than full conditional density estimation on the original outcome space. We show how to estimate this transform in practice and derive bounds on the conditional coverage gap, alongside volumetric and symmetric-difference bounds. We present known minimax-optimal conditional estimation techniques while also motivating the use of modern conditional density estimators, including Mixture Density Networks and Conditional Normalizing Flows. Finally, we empirically demonstrate on various datasets that our PIT-CP procedure matches or outperforms many state-of-the-art conformal prediction strategies with minimal effort and computational cost.
\end{abstract}

\medskip


\section*{Introduction}

Uncertainty quantification is a critical requirement for the safe and reliable adoption of machine learning models in high-stakes domains, especially as black-box deep neural networks increasingly become standard practice. Over the past two decades, conformal prediction \citep{vovk2005algorithmic} has emerged as a distribution-agnostic framework capable of producing prediction regions with finite-sample marginal coverage guarantees under very mild conditions, regardless of model misspecification. However, marginal validity controls the error rate only on average over the feature distribution, whereas many applications require \emph{conditional} coverage, namely validity for a specific individual with features $X = x$.

Unfortunately, finite-sample nontrivial conditional coverage is impossible in a fully distribution-free setting when the feature space $\mathcal{X}$ is infinite. Foundational impossibility results \citep{vovk2012conditional, lei2014distribution, barber2021limits} show that without structural assumptions on the data-generating process, such guarantees can only be obtained through trivial prediction regions in the outcome space $\mathcal{Y}$. To mitigate this limitation, existing work has considered binning or stratification strategies \citep{vovk2005algorithmic, papadopoulos2002inductive}, locally adaptive conformal methods \citep{papadopoulos2008normalized, lei2018distribution, papadopoulos2011reliable, han2022split, hore2023conformal, guan2023localized, gibbs2025conformal}, and Conformalized Quantile Regression (CQR) \citep{romano2019conformalized}. While these approaches improve local adaptivity, they can suffer from discretization effects, the curse of dimensionality, or too restrictive geometric shapes.

These difficulties become even more visible in multidimensional and multimodal outcome spaces, where scalar reductions such as the magnitude of a residual often lead to overly large hyperrectangular or elliptical regions. To address this, recent work has incorporated optimal transport tools \citep{thurin2025optimal} or generative modeling and density estimation into conformal prediction. For example, Conformal Prediction Region via Normalizing Flow Transformation (CONTRA) \citep{fang2025contra} seeks to map varying conditional distributions to a fixed base distribution, while methods such as Joint Adaptive Prediction Areas with Normalizing-Flows (JAPAN) \citep{english2025japan}, Highest Predictive Density (HPD) split \citep{izbicki2022cd}, and the CP\textsuperscript{2} framework \citep{plassier2024probabilistic} use estimated conditional distributions or densities to construct flexible multimodal prediction regions and study approximate conditional validity.

Despite the widespread adoption of these various methods, quantifying the conditional coverage gap remains challenging, and many results still rely on highly \textit{ad hoc} convergence assumptions. Density-based scores, though often efficient in volume, may also lose interpretability because the resulting prediction regions do not necessarily admit a simple description, especially in high dimensions where learning the full conditional distribution is both statistically and computationally demanding. To address this, we develop a unified framework built upon Split Conformal Prediction (SCP), showing that conditional validity is fundamentally equivalent to the construction of pivotal scores whose distribution does not depend on $X$. We show that the idea underlying HPD-split extends to a general monotone, optimal transport-driven, post-processing correction acting on an arbitrary base nonconformity score potentially derived from highly accurate point predictions. Moreover, this transformation preserves the original geometry and marginal behavior of the score, while also reducing the problem to one-dimensional conditional density estimation, even when the original outcome space is high-dimensional.

\subsection*{Main Contributions}

Building on recent work, we:
\begin{itemize}[itemsep=0pt]
    \item derive general Kolmogorov--Smirnov bounds on the conditional coverage gap for i.i.d. data, showing the equivalence between conditional coverage and distributional invariance of the nonconformity score with respect to the features (Section~\ref{sec:framework}).
    \item introduce a high-level post-processing correction that maps any base nonconformity score to an approximately invariant one while preserving its geometry and marginal coverage (Subsections~\ref{subsec:pit} and \ref{subsec:pit_plugin}).
    \item derive explicit finite-sample deterministic and high-probability bounds on the conditional coverage gap of our plug-in correction (Subsection~\ref{subsec:cond}), and additionally establish volumetric and symmetric-difference bounds (Subsection~\ref{subsec:vol}).
    \item derive explicit tight convergence rates for the conditional coverage gap under suitable smoothness assumptions, while also motivating the use of modern, state-of-the-art estimators such as Mixture Density Networks and Conditional Normalizing Flows in practical algorithms (Section~\ref{sec:algorithms}).
    \item illustrate and compare our method on both simulated and real-world data using our companion \href{https://pypi.org/project/pitcp/}{\texttt{pitcp}} Python package, highlighting the versatility, simplicity, and performance of the PIT-CP procedure (Section~\ref{sec:num}).
\end{itemize}

\section{Theoretical Framework} \label{sec:framework}

Before detailing our proposed methodology, we first warm up by recalling the standard split conformal prediction framework and its key assumptions. We then introduce structural assumptions under which pivotal scores formalize exact conditional coverage, and subsequently extend this perspective to approximate conditional coverage guarantees.

Throughout the following sections, we assume that all conditional distributions are well-defined whenever invoked. In particular, this is guaranteed for random variables taking values in Polish spaces.

\subsection{Split Conformal Prediction Setting} \label{subsec:setting}

Let $\mathcal{X}$ be a feature space and $\mathcal{Y}$ an outcome space. While $\mathcal{Y}$ is often a subset of $\mathbb{R}$ in univariate settings \citep{vovk2005algorithmic}, our framework accommodates multidimensional spaces ($\mathcal{Y} \subseteq \mathbb{R}^d$ with $d > 1$). We observe $\left\{ (X_i, Y_i) \right\}_{i=1}^n \subseteq \mathcal{X} \times \mathcal{Y}$ for $n \geq 1$, with marginal distribution $P_{XY}$. Distribution-free guarantees rely on the following assumption.

\begin{assumption}[Exchangeability] \label{ass:exchangeability}
The calibration data $\mathcal{D}_n \coloneqq \left\{ (X_i, Y_i) \right\}_{i=1}^n$ and the test point $(X_{n + 1}, Y_{n + 1})$ are exchangeable, \textit{i.e.}, the joint distribution of $(X_i, Y_i)_{i=1}^{n + 1}$ is invariant under permutations
\[
\forall \sigma \in \mathfrak{S}_{n + 1}, \, (X_1, Y_1), \dots, (X_{n + 1}, Y_{n + 1}) \overset{d}{=} (X_{\sigma(1)}, Y_{\sigma(1)}), \dots, (X_{\sigma(n + 1)}, Y_{\sigma(n + 1)}).
\]
\end{assumption}

Exchangeability is standard and holds in particular for i.i.d. data. In split conformal prediction, we typically use a measurable nonconformity score $s: \mathcal{X} \times \mathcal{Y} \to \mathbb{R}$ measuring disagreement between $y$ and the model prediction at $x$. To avoid ties, we require the following continuity assumption.

\begin{assumption}[Continuous Score Distribution] \label{ass:cont}
For $(X, Y) \sim P_{XY}$, the nonconformity score $S \coloneqq s(X, Y)$ has a continuous cumulative distribution function (CDF), and $\mathbb{P}(S_i = S_j) = 0$ for all $i \neq j$, where $S_i \coloneqq s(X_i, Y_i)$ for $i \in \llbracket n \rrbracket$.
\end{assumption}

\begin{remark}
If the CDF is not continuous, one can always add a small continuous i.i.d. noise term to obtain a continuous CDF.
\end{remark}

Given $\mathcal{D}_n$, for a new feature $X_{n + 1} = x$, we construct a prediction region with target coverage $1 - \alpha \in (0, 1)$ by computing the empirical quantile $\hat{q}_{1 - \alpha}$ as the $k$-th order statistic of $\{ S_i \}_{i=1}^n \cup \{ +\infty \}$ with $k = \lceil (n + 1)(1 - \alpha) \rceil$, and $+\infty$ if $k = n + 1$.

The conformal prediction region is then defined by
\[
\widehat{C}_{1 - \alpha}(x) \coloneqq \{ y \in \mathcal{Y} : s(x, y) \leq \hat{q}_{1 - \alpha} \}.
\]

By Assumption~\ref{ass:exchangeability}, we obtain marginal, distribution-free validity \citep{vovk2005algorithmic}
\begin{equation} \label{eq:marginal}
\mathbb{P}_{\mathcal{D}_{n + 1}}\left( Y_{n + 1} \in \widehat{C}_{1 - \alpha}(X_{n + 1}) \right) \geq 1 - \alpha,
\end{equation}
and under Assumption~\ref{ass:cont}
\begin{equation} \label{eq:exact_marginal}
1 - \alpha \leq \mathbb{P}_{\mathcal{D}_{n + 1}}\left( Y_{n + 1} \in \widehat{C}_{1 - \alpha}(X_{n + 1}) \right) \leq 1 - \alpha + \frac{1}{n + 1},
\end{equation}
where $\mathbb{P}_{\mathcal{D}_{n + 1}}$ denotes the marginalized probability over the calibration data and $(X_{n + 1}, Y_{n + 1})$.

These follow from the rank of $S_{n + 1}$ being uniformly distributed among the $n + 1$ scores under exchangeability. From this point onward, all subsequent results are to be understood as implicitly marginalized over $\mathcal{D}_n$ except when stated otherwise, even when not stated explicitly. This important point should be kept in mind, as notation can often be misleading.

\subsection{Pivotal Scores}

As marginal validity is often insufficient in practice, we instead aim for conditional coverage, namely $\mathbb{P}\left( Y_{n + 1} \in \widehat{C}_{1 - \alpha}(x) \mid X_{n + 1} = x \right) \geq 1 - \alpha$ for almost all $x \in \mathcal{X}$. However, as previously discussed, exact nontrivial conditional coverage is unattainable without additional structural assumptions on the data-generating distribution. Still, considering the idealized (oracle) setting in which the joint distribution $P_{XY}$ is known provides useful guidance for designing estimation strategies. This perspective motivates the introduction of scores whose distribution is invariant with respect to $X$.

\begin{definition}[Pivotal Score] \label{def:pivotal}
A nonconformity score $S \coloneqq s(X, Y)$ is pivotal given $(X, Y) \sim P_{XY}$ if its distribution does not depend on $X$, \textit{i.e.}, $S \indep X$.
\end{definition}

\begin{remark}
Equivalently, for $P_X$-almost all $x \in \mathcal{X}$, $\forall t \in \mathbb{R}, \, F_{S \mid X = x}(t) = F_S(t)$.
\end{remark}

We now strengthen the exchangeability assumption and suppose that the data are i.i.d. This requirement is needed because, under exchangeability alone, conditioning on $X_{n + 1}$ may change the joint distribution of the sample and affect both the calibration scores and the quantile $\hat{q}_{1 - \alpha}$. By contrast, under the i.i.d. assumption, if the score is pivotal, then the rank of $S_{n + 1}$ remains uniformly distributed even conditionally on given features $X_{n + 1}$, which leads to the conditional validity statement below.

\begin{assumption}[Independent and Identically Distributed Data] \label{ass:iid}
The calibration data $\mathcal{D}_n$ and $(X_{n + 1}, Y_{n + 1})$ are drawn i.i.d. from a joint probability distribution $P_{XY}$, \textit{i.e.},
\[
\forall i \in \llbracket n + 1 \rrbracket, \, (X_i, Y_i) \overset{\mathrm{i.i.d.}}{\sim} P_{XY}.
\]
\end{assumption}

\bigbreak

We now present the main theorem of this section, establishing the equivalence between almost-sure conditional coverage and the pivotality of the score. We recall that we restrict ourselves to the split conformal framework, where $\widehat{C}_{1 - \alpha}(x)$ given $x \in \mathcal{X}$ is computed according to Subsection~\ref{subsec:setting}.

\begin{theorem}[Conditional Coverage] \label{thm:cond_coverage}
If Assumption~\ref{ass:iid} holds, then $\mathrm{(1)}$ and $\mathrm{(2)}$ are equivalent, where:
\begin{enumerate}[label=(\arabic*), noitemsep]
    \item $\forall n \geq 1, \forall \alpha \in (0, 1), \, \mathbb{P}\left( Y_{n + 1} \in \widehat{C}_{1 - \alpha}(X_{n + 1}) \mid X_{n + 1} \right) = \mathbb{P}\left( Y_{n + 1} \in \widehat{C}_{1 - \alpha}(X_{n + 1}) \right) \quad P_X \text{-a.s.}$,
    \item The nonconformity score $s$ is pivotal, that is $S \indep X$.
\end{enumerate}
\end{theorem}

\bigbreak

The proof of Theorem~\ref{thm:cond_coverage} is provided in Appendix~\ref{pf:cond_coverage}.

\bigbreak

Therefore, pivotal scores transfer marginal validity to conditional validity, and the converse also holds with i.i.d. data, so that the two properties are in that sense essentially equivalent. From this perspective, conditional conformal prediction may be viewed as the problem of constructing nonconformity scores that are, at least approximately, independent of $X$. This also explains the impossibility results for distribution-free conditional coverage, as a score pivotal under all joint distributions on $(X, Y)$ is necessarily almost surely constant, rendering the prediction regions trivial.

\begin{remark}
It is important to note that the fact that the distribution of $S \coloneqq s(X, Y)$ is independent of $X$ does not imply that $\widehat{C}_{1 - \alpha}$ itself is independent of $X$. On the contrary, as illustrated in the framework of \emph{aware} algorithmic fairness, achieving independence from $X$ does in fact require explicitly incorporating $X$ into the construction.
\end{remark}

\subsection{Approximate Independence Bounds}

Since exact independence is typically unattainable without knowledge of the data-generating process, we consider approximate independence and quantify the resulting conditional coverage gap, acting uniformly over all miscoverage levels $\alpha \in (0, 1)$. Given features $x \in \mathcal{X}$, we use the Kolmogorov--Smirnov distance between $F_{S \mid X = x}$ and $F_S$, as in \cite{ding2023class} and \cite{gao2024adjusting}.

\begin{definition}[Kolmogorov--Smirnov distance] \label{def:ks_distance}
Let $F$ and $G$ be CDFs on $\mathbb{R}$. The Kolmogorov--Smirnov distance between $F$ and $G$ is defined by
\[
d_{KS}(F, G) \coloneqq \sup_{t \in \mathbb{R}}\left\vert F(t) - G(t) \right\vert.
\]
\end{definition}

\begin{definition}[Conditional Coverage Gap] \label{def:coverage_gap}
Upon observing features $X_{n + 1} = x \in \mathcal{X}$, the (uniform) conditional coverage gap (with implicit dependence on $n$) is defined by
\[
\Delta(x) \coloneqq \sup_{\alpha \in (0, 1)} \left\vert \mathbb{P}\left( Y_{n + 1} \in \widehat{C}_{1 - \alpha}(x) \mid X_{n + 1} = x \right) - \mathbb{P}\left(Y_{n + 1} \in \widehat{C}_{1 - \alpha}(X_{n + 1}) \right) \right\vert.
\]
\end{definition}

\begin{lemma}[Kolmogorov--Smirnov Conditional Coverage Gap Bound] \label{lem:ks_bound}
Suppose Assumption~\ref{ass:iid} holds. Then, the conditional coverage gap is (deterministically) bounded by
\[
\Delta(x) \leq d_{KS}(F_{S \mid X = x}, F_S).
\]
\end{lemma}

\bigbreak

The proof of Lemma~\ref{lem:ks_bound} is provided in Appendix~\ref{pf:ks_bound}.

\begin{remark}
The bound of Lemma~\ref{lem:ks_bound} is asymptotically tight for the full conditional coverage gap under Assumption~\ref{ass:cont}. Indeed, the empirical quantiles converge almost surely to the true population quantiles, so the \textit{supremum} over $\alpha \in (0, 1)$ yields a gap that converges exactly to the Kolmogorov--Smirnov distance $d_{KS}(F_{S \mid X = x}, F_S)$.
\end{remark}

\begin{remark}
The implicit dependence of the conditional coverage gap on $n$ is resolved by Lemma~\ref{lem:ks_bound}, which ensures a universal bound independent of the calibration size.
\end{remark}

\bigbreak

This result already suggests a natural approach: any procedure that enforces approximate distributional invariance of the scores across values of $X$ in a Kolmogorov--Smirnov sense automatically controls the conditional coverage gap uniformly over all miscoverage levels $\alpha \in (0, 1)$.

\begin{illustration} \label{il:exp}
Consider $X \sim \mathrm{Unif}(0, 1)$ and heteroscedastic noise $Y \mid X = x \sim \mathrm{Lap}\left( 0, \frac{1}{x+1} \right)$. Using the absolute residual $s(x, y) = \vert y \vert$, one has $S \mid X = x \sim \mathrm{Exp}(x + 1)$, and thus for all $t > 0$
\[
F_{S \mid X = x}(t) = 1 - e^{-(x+1)t}, \quad F_S(t) = \int_0^1 \left( 1 - e^{-(x+1)t} \right) \, dx = 1 - \frac{e^{-t} - e^{-2t}}{t}.
\]

Figure~\ref{fig:ks} visualizes this envelope, highlighting the Kolmogorov--Smirnov distance which explicitly bounds the conditional coverage gap.
\end{illustration}

\begin{figure}[H]
    \centering
    \includegraphics[width=0.6\linewidth]{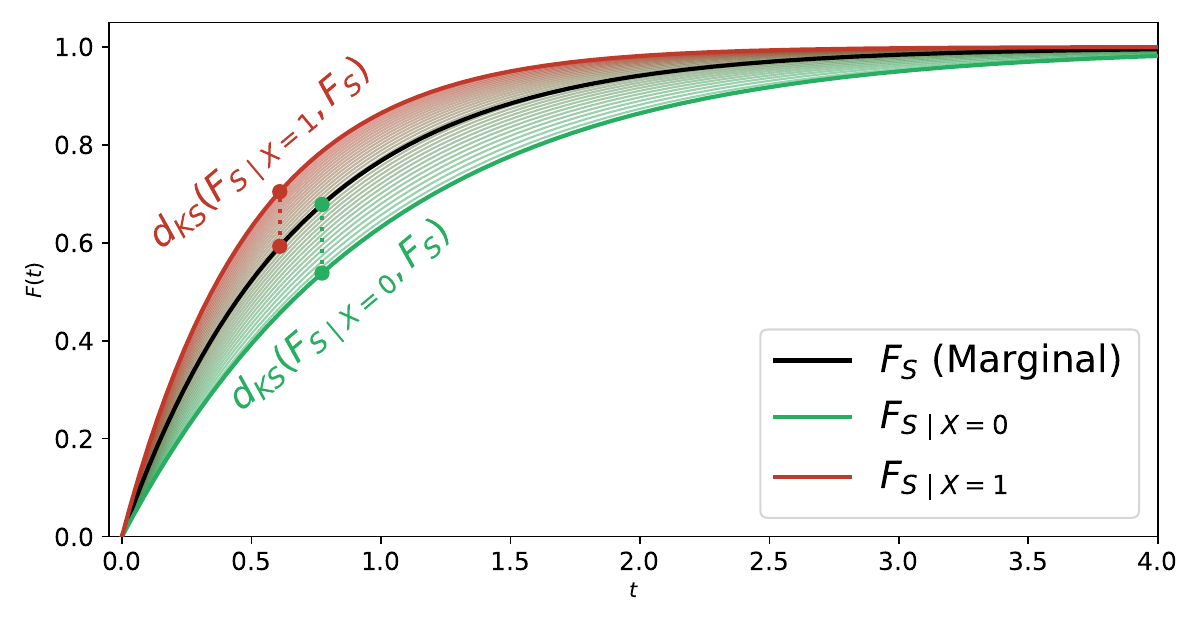}
    \caption{Illustration of the Kolmogorov--Smirnov distance under the heteroscedastic noise model of Illustration~\ref{il:exp}. Multiple conditional CDFs are plotted with color gradient, uniformly ranging from $F_{S \mid X = 0}$ to $F_{S \mid X = 1}$.}
    \label{fig:ks}
\end{figure}

\section{Related Work}

In practice, through the lens of Definition~\ref{def:pivotal}, many known methods can be interpreted as trying to make the score distribution less dependent on $X$, progressing from standardization techniques to full distributional normalization.

\paragraph{Locally Adaptive Scaling.}

A classical strategy for heteroscedastic regression is to standardize residuals using estimated conditional location and scale, popularized by \cite{papadopoulos2008normalized} and \cite{lei2018distribution} among others. In this setting, assuming a univariate response, the score is defined by
\[
s_\mathrm{LS}(x, y) = \frac{\vert y - \hat{\mu}(x) \vert}{\hat{\sigma}(x)},
\]
where $\hat{\mu}$ and $\hat{\sigma}$ typically estimate the conditional mean and standard deviation.

Under an exact location-scale model, for instance $Y \mid X \sim \mathcal{N}\left( \mu(X), \sigma(X)^2 \right)$ with $\hat{\mu} = \mu$ and $\hat{\sigma} = \sigma$, the score is independent of $X$ up to a half-normal distribution, hence pivotal. More generally, this method normalizes the first two conditional moments and can improve coverage under heteroscedasticity, but it does not correct higher-order features such as skewness or multimodality. These types of location-scale models have also been studied in other nonparametric approaches such as that of \cite{matabuena2024conformal}.

\paragraph{Cumulative Distribution Approaches.}

A stronger idea is to normalize the full conditional distribution. For real-valued responses, Distributional Conformal Prediction (DCP) \citep{chernozhukov2021distributional} uses an estimated conditional cumulative distribution function and defines
\[
s_\mathrm{DCP}(x, y) = \left\vert \widehat{F}_{Y \mid X = x}(y) - \frac{1}{2} \right\vert.
\]

The rationale is rooted in the Probability Integral Transform (PIT): if $Y \mid X$ is continuous, then $F_{Y \mid X}(Y)$ is uniformly distributed conditionally on $X$. Hence, accurate estimation of the conditional cumulative distribution function $\widehat{F}_{Y \mid X}$ makes the transformed score approximately invariant with respect to $X$. Compared with locally adaptive scaling, DCP targets the full conditional distribution rather than only its first two moments, and unlike CQR, it also does not require training separate models for each miscoverage level.

\paragraph{Normalizing Flow Approaches.}

For multidimensional outcomes, recent work uses Conditional Normalizing Flows to map conditional distributions to a fixed latent distribution, typically chosen as a multivariate normal distribution purely for convenience and numerical stability. In Conformal Prediction Region via Normalizing Flow Transformation (CONTRA) \citep{fang2025contra}, the procedure learns a parametric diffeomorphism $f_\theta$ sending a pair $(x, y) \in \mathcal{X} \times \mathcal{Y}$ to
\[
z = f_\theta(y \mid x),
\]
and defines the nonconformity score quantile by the Euclidean norm in the latent space
\[
s_\mathrm{CONTRA}(x, y) = \Vert f_\theta(y \mid x) \Vert.
\]

If the flow successfully transports, the conditional distribution $P_{Y \mid X = x}$ close to a feature-independent base measure for all $x \in \mathcal{X}$, then the induced latent score varies less across $x$. For instance, the resulting score is close to $\sqrt{\chi^2(d)}$ in distribution when the latent space is equipped with the standard normal distribution in $\mathbb{R}^d$. Although the main appeal of this method is driven by the geometrical properties of the resulting conformal prediction regions, approximate invariance may arise as a by-product and thus approximate conditional validity, as the authors briefly mention through Theorem~2.7 of \cite{colombo2024normalizing}.

\paragraph{Density-Based Frameworks.}

Another idea is to use estimated conditional densities directly to approach smallest-volume prediction regions. In methods such as HPD-split \citep{izbicki2022cd} and Joint Adaptive Prediction Areas with Normalizing Flows (JAPAN) \citep{english2025japan}, one first fits a generative model to approximate the conditional distribution of the outcome given the features, yielding a base nonconformity score
\[
s_\mathrm{density}(x, y) = -p_\theta(y \mid x).
\]

However, such density scores are not generally pivotal, since their distribution still depends on $X$. HPD-split addresses this by estimating the conditional distribution of the density scores and applying a conditional quantile transform, thereby stabilizing their distribution while preserving split conformal marginal validity. JAPAN instead works directly with the estimated density with no intention of achieving conditional validity.

Lastly, the CP\textsuperscript{2} method \citep{plassier2024probabilistic} also relies on an estimated conditional distribution, using a score induced by a Mixture Density Network to produce interpretable prediction regions, typically unions of ellipses. These density-based methods are closely related to our viewpoint: they seek approximate conditional validity through distributional normalization, but usually remain tied to a specific score construction and often require modeling the full conditional distribution of $Y \mid X$, rather than a systematic post-processing correction.

\section{Theoretical Construction and Guarantees}

While methods such as DCP, CONTRA, HPD-split, and CP\textsuperscript{2} use generative modeling to encourage distributional invariance, they remain tied to specific score constructions and often require estimating the full conditional distribution of the response variable. In multidimensional outcome spaces, this can be statistically difficult and computationally costly, and thus motivates a more general principle that acts directly on an arbitrary base score, including problem-specific scores derived from pre-trained predictive models whose geometry or interpretability one may wish to retain.

We therefore introduce a monotone post-processing correction that maps any base nonconformity score to a pivotal score while preserving its properties. Conceptually, this extends the calibration idea underlying HPD-split, but separates the correction from the initial score definition.

\bigbreak

To formalize this construction, let $s$ denote an arbitrary measurable base nonconformity score, such as an absolute residual, a neural network distance, or a negative $\log$-probability score, with no special constraints. To simplify the notation, in the following sections, also let $(X, Y) = (X_{n + 1}, Y_{n + 1})$ denote the test point, which is independent of the calibration data $\mathcal{D}_n$ under Assumption~\ref{ass:iid}.

\subsection{The PIT Correction} \label{subsec:pit}

In general, the distribution of $S \coloneqq s(X, Y)$ depends on $X$ and therefore does not by itself guarantee conditional coverage. We address this by composing the base score with the nondecreasing conditional CDF of the induced score, thereby preserving the shape of the conformal prediction regions.

\begin{definition}[PIT Correction] \label{def:pit_correction}
Given a base score $s$, the PIT-corrected score $\tilde{s}$ is defined by
\[
\forall (x, y) \in \mathcal{X} \times \mathcal{Y}, \, \tilde{s}(x, y) \coloneqq \mathbb{P}\left( s(x, Y) \leq s(x, y) \mid X = x \right) = F_{S \mid X = x}\left( s(x, y) \right).
\]
\end{definition}

\begin{example}[Locally Adaptive Scaling]
Suppose $Y \mid X = x \sim \mathcal{N}\left( \mu(x), \sigma(x)^2 \right)$ and consider the absolute residual score $s(x, y) = \vert y - \mu(x) \vert$. Conditionally on $X = x$, the score follows a scaled half-normal distribution, yielding the PIT-corrected score
\[
\tilde{s}(x, y) = \mathbb{P}\left( \frac{\vert  Y - \mu(x) \vert}{\sigma(x)} \leq \frac{\vert y - \mu(x) \vert}{\sigma(x)} \mid X = x \right) = 2 \Phi\left( \frac{\vert y - \mu(x) \vert}{\sigma(x)} \right) - 1,
\]
where $\Phi$ is the standard normal cumulative distribution function. As such, the PIT correction behaves exactly as the classical locally adaptive scaling score of \cite{lei2018distribution}, up to a strictly increasing transformation.
\end{example}

\bigbreak

Taking its name from the Probability Integral Transform, the main property of the PIT correction is its monotonicity. For any fixed $x \in \mathcal{X}$, the map
\[
t \mapsto \mathbb{P}\left( s(x, Y) \leq t \mid X = x \right) = F_{S \mid X = x}(t)
\]
is a cumulative distribution function, hence nondecreasing in $t$.

As formalized by \cite{hanselle2025conformal}, split conformal prediction depends only on the ordering of candidate labels. Since the PIT correction is a nondecreasing transform of $s$, it preserves the ranking of the data. More formally, the sublevel sets of $\tilde{s}$ coincide with those of $s$ after reparameterizing the threshold: for any fixed $x \in \mathcal{X}$ and any threshold $\tilde{t} \in (0, 1)$, there exists a corresponding threshold $t$ such that
\[
\{ y \in \mathcal{Y} : \tilde{s}(x, y) \leq \tilde{t} \} = \{ y \in \mathcal{Y} : s(x, y) \leq t \}.
\]

Thus, the PIT correction does not alter the geometry of the initial prediction regions: it only reindexes the same nested family of sets in a probability scale. It acts as a post-processing transformation that completely decouples the initial predictive modeling task from that of achieving conditional coverage. Also, unlike \cite{plassier2025rectifying}, it does not target a single quantile, imposes no Lipschitz assumptions in general, and instead acts uniformly over all quantiles by uniformizing the full conditional distribution. As such, the prediction regions can be readily obtained for all miscoverage levels $\alpha \in (0, 1)$, without requiring separate models.

Beyond preserving geometry, under a continuity assumption, the transformed score is also distributionally invariant, leading to the result below.

\begin{assumption}[Continuous Conditional CDF] \label{ass:cont_cond}
For all $x \in \mathcal{X}$, the conditional CDF $F_{S \mid X = x}: t \mapsto \mathbb{P}\left( s(x, Y) \leq t \mid X = x \right)$ is continuous.
\end{assumption}

\begin{proposition}[Invariance of the CDF-Transformed Score] \label{prop:pit_invariance}
Suppose Assumption~\ref{ass:cont_cond} holds. Then, the PIT-corrected score $\widetilde{S} \coloneqq \tilde{s}(X, Y)$ as in Definition~\ref{def:pit_correction} is pivotal and uniformly distributed on $(0, 1)$.
\end{proposition}

\bigbreak

The proof of Proposition~\ref{prop:pit_invariance} is provided in Appendix~\ref{pf:pit_invariance}.

\bigbreak

Under the standard continuity assumption, the PIT correction at each $x \in \mathcal{X}$ coincides with the one-dimensional monotone optimal transport map from $P_{S \mid X = x}$ to $\mathrm{Unif}(0, 1)$. As a nonconformity score, it additionally satisfies the marginal coverage property without requiring any distributional assumptions. Moreover, it is canonically determined by the conditional law of the score, since any monotone map sending $P_{S \mid X = x}$ to a common target distribution with a continuous cumulative distribution function must be of the form $\psi \circ F_{S \mid X = x}$ for some strictly increasing function $\psi$. This viewpoint also clarifies why the source distribution must be atomless. Indeed, a deterministic transport map cannot spread an atom of $P_{S \mid X = x}$ over an interval, so a score distribution with point masses cannot be mapped by a monotone transport to the diffuse distribution $\mathrm{Unif}(0, 1)$. This is consistent with the optimal transport perspective underlying Brenier's theorem \citep{brenier1991polar}. From this perspective, our correction parallels recent post-processing techniques for demographic parity in regression, which similarly construct transport maps to enforce distributional invariance with respect to a sensitive attribute \citep{chzhen2020fair}.

\begin{figure}[H]
    \centering
    \includegraphics[width=0.6\linewidth, trim=7 7 7 7, clip]{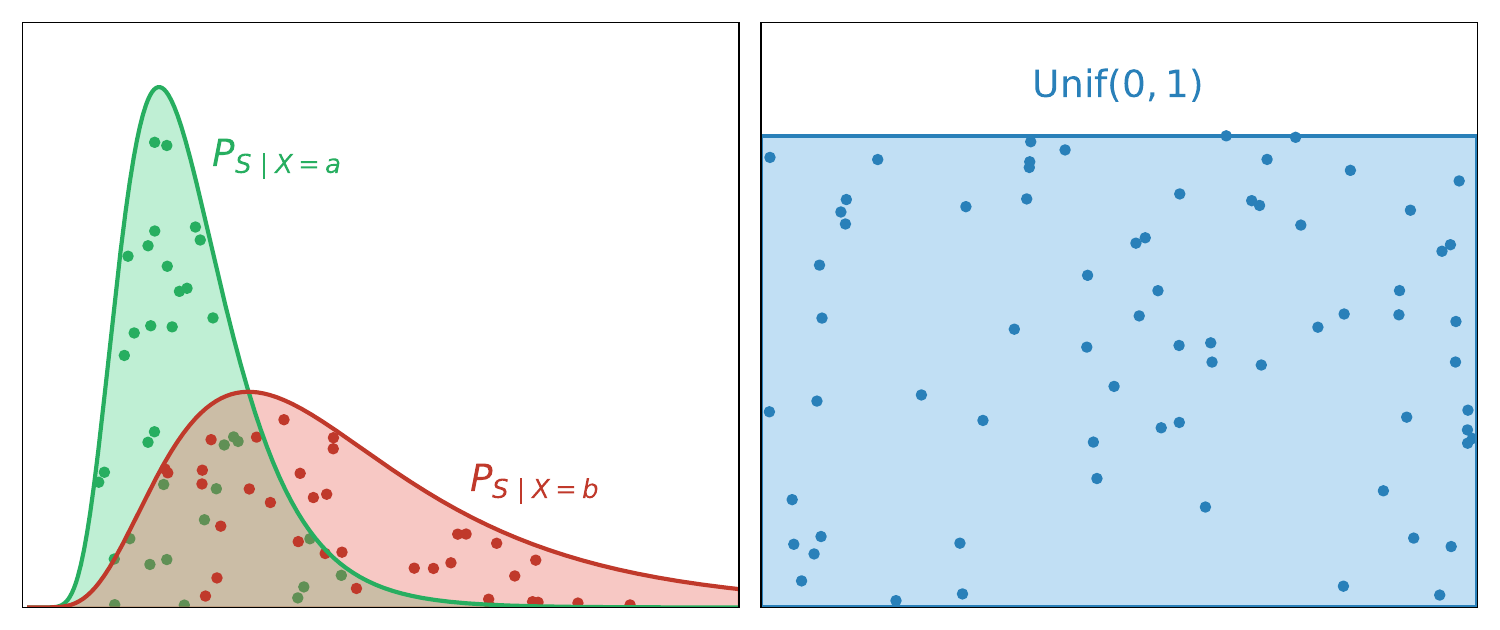}
    \caption{Left: conditional score distribution $P_{S \mid X = x}$ and its density; right: corresponding monotone transport to a common $\mathrm{Unif}(0, 1)$ distribution.}
    \label{fig:ot}
\end{figure}

\subsection{Plug-in Estimation} \label{subsec:pit_plugin}

In practice, the true conditional distribution of the base score is unknown. To approximate the ideal transform, we require an estimator indexed by $x \in \mathcal{X}$, either
\[
\widehat{P}_{S \mid X = x} \qquad \text{or} \qquad \widehat{P}_{Y \mid X = x}.
\]

Either is sufficient, since $\widehat{P}_{Y \mid X = x}$ induces a conditional pushforward distribution on $S \mid X = x$. Estimating $\widehat{P}_{S \mid X = x}$ is often preferable because the base score is always real-valued, regardless of the dimension of $\mathcal{Y}$. Intuitively, if the base score $s(x, y)$ incorporates a point prediction $\hat{f}(x)$, modeling the conditional distribution of the score reframes the problem from asking ``what is the full distribution of the target?'' to ``how uncertain is our prediction?''. As such, it is expected for the residuals or scores to be better behaved and their distribution easier to estimate. However, in settings such as HPD-split, the score is already defined through an estimate of $P_{Y \mid X = x}$, so a separate model for $P_{S \mid X}$ may be unnecessary. Our bounds apply in either case, up to replacing the Kolmogorov--Smirnov distance by the total variation distance, as made evident in Remark~\ref{rmk:contraction}.

For the moment, we assume that one of these two estimators is available, fixed, and independent of the calibration data $\mathcal{D}_n$, corresponding to the standard split between model fitting and conformal calibration. Consequently, all probabilistic statements in the following section are implicitly understood to be conditional on this fixed estimator. Under this assumption, the most natural estimator is the plug-in PIT-corrected score, defined by
\begin{equation} \label{eq:plugin_score}
\hat{s}(x, y) \coloneqq \widehat{F}_{S \mid X = x}\left( s(x, y) \right),
\end{equation}
and we denote by $\widehat{\Delta}(x)$ its corresponding conditional coverage gap for a given point $x \in \mathcal{X}$.

\subsection{Conditional Coverage Gap Bounds} \label{subsec:cond}

Because the true conditional distribution is replaced by an estimate, the score $\widehat{S} \coloneqq \hat{s}(X, Y)$ is typically no longer exactly uniformly distributed. Its conditional coverage error can still be bounded using the Kolmogorov--Smirnov distance. Since this supremum-based quantity is not especially convenient for optimization, we also relate it to the total variation distance and then to the Kullback--Leibler divergence via Pinsker's inequality. This is especially insightful as likelihood-based estimators are often trained through objectives involving Kullback--Leibler divergence, and thus indirectly optimize bounds on the conditional coverage gap. We recall the two discrepancy measures below.

\begin{definition}[Total Variation Distance] \label{def:tv_distance}
Let $P$ and $Q$ be two probability measures on a measurable space $(\Omega, \mathcal{F})$. The total variation distance between $P$ and $Q$ is defined by
\[
d_{TV}(P, Q) \coloneqq \sup_{A \in \mathcal{F}}\left\vert P(A) - Q(A) \right\vert.
\]
\end{definition}

\begin{definition}[Kullback--Leibler Divergence] \label{def:kl_divergence}
Let $P$ and $Q$ be two probability measures on a measurable space. The Kullback--Leibler divergence of $P$ with respect to $Q$ is defined by
\[
D_{KL}(P \parallel Q) \coloneqq \begin{cases} 
\int \log\left(\frac{dP}{dQ}\right) dP &\text{if } P \ll Q, \\
+\infty &\text{otherwise}.
\end{cases}
\]
\end{definition}

\bigbreak

To state the following lemma and theorem, which link the plug-in estimator more directly to the true conditional distribution than the implicit bound of Lemma~\ref{lem:ks_bound}, we impose the continuity condition below.

\begin{assumption}[Continuous Plug-in Conditional CDF] \label{ass:plugin_cont_cond}
For all $x \in \mathcal{X}$, the plug-in conditional cumulative distribution function $\widehat{F}_{S \mid X = x}$ is continuous.
\end{assumption}

\begin{lemma}[Convenient Upper Bounds] \label{lem:plugin_upper}
Suppose Assumptions~\ref{ass:iid} and \ref{ass:plugin_cont_cond} hold. Then, for all $x \in \mathcal{X}$
\begin{align}
d_{KS}(F_{\widehat{S} \mid X = x}, F_{\widehat{S}}) &\leq d_{KS}\left( F_{S \mid X = x}, \widehat{F}_{S \mid X = x} \right) + \mathbb{E}\left[ d_{KS}\left( F_{S \mid X}, \widehat{F}_{S \mid X} \right) \right], \label{eq:ks_bound} \\
d_{KS}(F_{\widehat{S} \mid X = x}, F_{\widehat{S}})^2 &\leq D_{KL}\left( P_{S \mid X = x} \parallel \widehat{P}_{S \mid X = x} \right) + \mathbb{E}\left[ D_{KL}\left( P_{S \mid X} \parallel \widehat{P}_{S \mid X} \right) \right]. \label{eq:kl_bound}
\end{align}
\end{lemma}

\begin{theorem}[Deterministic Plug-in Conditional Coverage Bounds]  \label{thm:plugin_bounds}
Suppose Assumptions~\ref{ass:iid} and \ref{ass:plugin_cont_cond} hold. Then, for all $x \in \mathcal{X}$, the conditional coverage gap satisfies
\begin{align*}
\widehat{\Delta}(x) &\leq d_{KS}\left( F_{S \mid X = x}, \widehat{F}_{S \mid X = x} \right) + \mathbb{E}\left[ d_{KS}\left( F_{S \mid X}, \widehat{F}_{S \mid X} \right) \right] \\
\widehat{\Delta}(x)^2 &\leq D_{KL}\left( P_{S \mid X = x} \parallel \widehat{P}_{S \mid X = x} \right) + \mathbb{E}\left[ D_{KL}\left( P_{S \mid X} \parallel \widehat{P}_{S \mid X} \right) \right].
\end{align*}
\end{theorem}

\bigbreak

The proofs of Lemma~\ref{lem:plugin_upper} and Theorem~\ref{thm:plugin_bounds} are provided in Appendix~\ref{pf:plugin_upper} and \ref{pf:plugin_bounds}, respectively.

\begin{remark}
The same inequality with the reverse Kullback--Leibler divergences also holds. As stated here with the forward Kullback--Leibler divergence, the bound is informative only when the estimated conditional distribution assigns sufficient mass wherever the true conditional distribution does, and thus heavier tails.
\end{remark}

\bigbreak

Besides a local error term, Equations~\eqref{eq:ks_bound} and~\eqref{eq:kl_bound} also contain an average error over the feature space, which propagates to the bounds in Theorem~\ref{thm:plugin_bounds}. This reflects the fact that although split conformal prediction provides global marginal validity, under-coverage on some regions of $\mathcal{X}$ must be compensated by over-coverage elsewhere, and conversely. This observation also aligns with the findings of \cite{braun2025conditional}.

\bigbreak

We also obtain $L^1$, $L^2$, as well as high-probability versions. This is especially useful because it connects conditional coverage control to standard maximum likelihood training objectives through the expected discrepancy between true and estimated conditional laws.

\begin{corollary}[Probabilistic Conditional Coverage Bound] \label{cor:exp_prob_bounds}
Suppose Assumptions~\ref{ass:iid} and \ref{ass:plugin_cont_cond} hold. Then, for any tolerance margin $\delta > 0$, the conditional coverage gap satisfies
\begin{equation} \label{eq:l1_l2_gap}
\mathbb{E}\left[ \widehat{\Delta}(X) \right] \leq 2\mathbb{E}\left[ d_{KS}\left( F_{S \mid X}, \widehat{F}_{S \mid X} \right) \right], \quad \mathbb{E}\left[ \widehat{\Delta}(X)^2 \right] \leq 2\mathbb{E}\left[ D_{KL}\left( P_{S \mid X} \parallel \widehat{P}_{S \mid X} \right) \right],
\end{equation}
\[
\mathbb{P}\left( \widehat{\Delta}(X) \geq \delta \right) \leq \frac{2\mathbb{E}\left[ d_{KS}\left( F_{S \mid X}, \widehat{F}_{S \mid X} \right) \right]}{\delta} \wedge \frac{2\mathbb{E}\left[ D_{KL}\left( P_{S \mid X} \parallel \widehat{P}_{S \mid X} \right) \right]}{\delta^2}.
\]
\end{corollary}

\bigbreak

The proof of Corollary~\ref{cor:exp_prob_bounds} is provided in Appendix~\ref{pf:exp_prob_bounds}.

\bigbreak

Taken together, Theorem~\ref{thm:plugin_bounds} and Corollary~\ref{cor:exp_prob_bounds} show that controlling the Kolmogorov--Smirnov distance or the forward Kullback--Leibler divergence probabilistically bounds the conditional coverage gap. Moreover, since many density estimators are trained by Maximum Likelihood Estimation (MLE, see Subsection~\ref{subsec:lik}), the forward Kullback--Leibler formulation gives a direct theoretical link between standard likelihood objectives and conditional validity. If the estimator $\widehat{P}_{S \mid X}$ satisfies
\[
\mathbb{E}\left[ D_{KL}\left( P_{S \mid X} \parallel \widehat{P}_{S \mid X} \right) \right] = O_P(\rho_N),
\]
for some sequence $\rho_N \to 0$ where $N$ denotes a training size used in the computation of $\widehat{P}_{S \mid X}$, then Corollary~\ref{cor:exp_prob_bounds} implies that $\widehat{\Delta}(X) = O_P(\sqrt{\rho_N})$, and thus converges in probability to zero, asymptotically achieving conditional validity with a known rate.

\begin{remark} \label{rmk:contraction}
Using the contraction of total variation distance under measurable mappings, both Theorem~\ref{thm:plugin_bounds} and Corollary~\ref{cor:exp_prob_bounds} admit analogous versions when replacing the Kolmogorov--Smirnov distance with the total variation distance and replacing $\widehat{P}_{S \mid X = x}$ with $\widehat{P}_{Y \mid X = x}$. Therefore, these bounds remain applicable in that setting with only minor modifications. Indeed, defining the measurable map $s_x : y \mapsto s(x, y)$, we have for all $x \in \mathcal{X}$
\begin{align*}
d_{KS}\left( F_{S \mid X = x}, \widehat{F}_{S \mid X = x} \right)
&\leq d_{TV}\left( P_{S \mid X = x}, \widehat{P}_{S \mid X = x} \right) \\
&= d_{TV}\left( (s_x)_\sharp P_{Y \mid X = x}, (s_x)_\sharp \widehat{P}_{Y \mid X = x} \right) \\
&\leq d_{TV}\left( P_{Y \mid X = x}, \widehat{P}_{Y \mid X = x} \right).
\end{align*}
\end{remark}

\subsection{Volumetric Bounds} \label{subsec:vol}

Of course, achieving conditional coverage is all but meaningless on its own, since trivial conformal prediction regions also attain a zero conditional coverage gap. Thus, we now compare the learned conformal prediction region with the oracle conditional prediction region generated by the original base score $s$. This makes the geometric meaning of a small conditional coverage gap explicit: if the transformed score is nearly pivotal, then the conformal prediction region must remain close, in a set-inclusion sense, to the oracle conditional sublevel set associated with the true conditional distribution, which we formally introduce below. To this end, we now require the conditional distributions of the score to be strictly increasing on their support.

\begin{assumption} \label{ass:increasing}
For all $x \in \mathcal{X}$, the conditional cumulative distribution function $F_{S \mid X = x}$ is strictly increasing on the interval $\left\{ t \in \mathbb{R} : 0 < F_{S \mid X = x}(t) < 1 \right\}$.
\end{assumption}

\begin{definition}[Oracle Conditional Prediction Region] \label{def:oracle_region}
Suppose Assumptions~\ref{ass:cont_cond} and \ref{ass:increasing} hold. Given a base nonconformity score $s$ and target miscoverage level $\beta \in (0, 1)$, the oracle conditional quantile at $x \in \mathcal{X}$ is uniquely defined by
\[
q_{1 - \beta}^*(x) \coloneqq  F_{S \mid X = x}^{-1}(1 - \beta),
\]
and the corresponding oracle conditional prediction region is defined by
\[
C_{1 - \beta}^*(x) \coloneqq \left\{ y \in \mathcal{Y} : s(x, y) \leq q_{1 - \beta}^*(x) \right\}.
\]
\end{definition}

\begin{remark}
We extend notation of Definition~\ref{def:oracle_region} so that $q_{1 - \beta}^*$ and $C_{1 - \beta}^*$ are defined for $\beta \in \mathbb{R}$, simply by truncating $\widetilde{\beta} = (\beta \vee 0) \wedge 1$, with $q_0^*(x) = -\infty$ and $q_1^*(x) = +\infty$ for all $x \in \mathcal{X}$.
\end{remark}

By construction, $C_{1 - \beta}^*(x)$ is the oracle element of the nested family of sublevel sets generated by $s$ having conditional coverage level $1 - \beta$. In other words, it is the prediction region one would obtain if the true conditional distribution of the base score were known.

\begin{lemma}[Oracle Inclusion Bounds] \label{lem:oracle_inclusion}
Suppose Assumptions~\ref{ass:iid}, \ref{ass:cont_cond} and \ref{ass:increasing} hold, and that ties in $\{ \widehat{S}_i \}_{i=1}^n$ occur with zero probability. For $\delta \in (0, 1)$, let
\[
\widehat{L}(x, n, \delta) = d_{KS}(F_{\widehat{S} \mid X = x}, F_{\widehat{S}}) + \sqrt{\frac{\log(2/\delta)}{2n}} + \frac{2}{n}.
\]

Then, for any miscoverage level $\alpha \in (0, 1)$, with probability at least $1 - \delta$ over $\mathcal{D}_n$, for all $x \in \mathcal{X}$
\[
C^*_{1 - \alpha - \widehat{L}(x, n, \delta)}(x) \subseteq \widehat{C}_{1 - \alpha}(x) \subseteq C^*_{1 - \alpha + \widehat{L}(x, n, \delta)}(x)
\]
\end{lemma}

\bigbreak

The proof of Lemma~\ref{lem:oracle_inclusion} is provided in Appendix~\ref{pf:oracle_inclusion}.

\bigbreak

Lemma~\ref{lem:oracle_inclusion} shows that a small conditional coverage gap yields more than approximate validity. Since all sets are sublevel sets of the same base score $s(x, \cdot)$, they have the same underlying geometry. The transform $\hat{s}$ does not change shape; it only selects a different level in the same nested family. In particular, the set inclusions imply that with probability at least $1 - \delta$ over $\mathcal{D}_n$
\[
\mathrm{Vol}\left( C^*_{1 - \alpha - \widehat{L}(x, n, \delta)}(x) \right) \leq \mathrm{Vol}\left( \widehat{C}_{1 - \alpha}(x) \right) \leq \mathrm{Vol}\left( C^*_{1 - \alpha + \widehat{L}(x, n, \delta)}(x) \right).
\]

Building on Lemma~\ref{lem:oracle_inclusion}, we now show that small conditional coverage gaps also entail asymptotic recovery of the oracle conditional region in symmetric difference, as detailed in Theorem~\ref{thm:sym_diff} below.

\begin{theorem}[Symmetric Difference Bound] \label{thm:sym_diff}
Suppose Assumptions~\ref{ass:iid}, \ref{ass:cont_cond}, \ref{ass:plugin_cont_cond}, and \ref{ass:increasing} hold, and that ties in $\{ \widehat{S}_i \}_{i=1}^n$ occur with zero probability. Then, for any miscoverage level $\alpha \in (0, 1)$, the symmetric difference satisfies
\begin{align*}
\mathbb{P}\left( Y \in \widehat{C}_{1 - \alpha}(X) \symdif C^*_{1 - \alpha}(X) \right) &\leq \inf_{\delta \in (0, 1)} \left\{ 2\mathbb{E}\left[ \widehat{L}(X, n, \delta) \right] + \delta \right\} \\
&\leq 4\mathbb{E}\left[ d_{KS}\left( F_{S \mid X}, \widehat{F}_{S \mid X} \right) \right] + O\left( \sqrt{\frac{\log(n)}{n}} \right) \\
&\leq 3\mathbb{E}\left[ D_{KL}\left( P_{S \mid X} \parallel \widehat{P}_{S \mid X} \right) \right]^{1/2} + O\left( \sqrt{\frac{\log(n)}{n}} \right).
\end{align*}
\end{theorem}

\bigbreak

The proof of Theorem~\ref{thm:sym_diff} is provided in Appendix~\ref{pf:sym_diff}.

\bigbreak

These results guarantee that the geometry of the conformal prediction region is controlled by the conditional coverage gap in a way akin to Theorem~\ref{thm:plugin_bounds}. As such, small errors imply closeness in the size of the resulting prediction region, whether it be in probability or in volume with respect to the Lebesgue measure. Specifically, Theorem~\ref{thm:sym_diff} bounds the symmetric difference directly through the expected forward Kullback--Leibler divergence, further reinforcing the link with standard MLE objectives discussed in Subsection~\ref{subsec:lik}.

\section{Estimation Algorithms} \label{sec:algorithms}

We now specialize our framework to the estimation of the conditional distribution of a fixed base score $S \coloneqq s(X, Y)$. This reduces the problem to one-dimensional conditional density estimation in $S$, regardless of the dimensionality of $\mathcal{Y}$, which is a key practical advantage over methods requiring estimation of the full conditional distribution.

Classical nonparametric methods such as conditional kernel density estimation \citep{hansen2004nonparametric} can in principle be applied, but they are sensitive to bandwidth selection, scale poorly with the dimension of the feature space $\mathcal{X}$, and require access to the full dataset at prediction time. While kernel-based procedures can nevertheless achieve minimax-optimal rates under suitable smoothness assumptions \citep{li2022minimax}, they remain of limited practical use in high-dimensional settings. We therefore also focus on more recent estimation approaches, and in particular on likelihood-based estimators that align with our Kullback--Leibler bounds, namely Mixture Density Networks and Conditional Normalizing Flows. More generally, our framework does not restrict the choice of conditional density estimator, and other parametric or nonparametric approaches could also be considered.

Furthermore, we emphasize that when the corrected calibration scores are stored, a key advantage of our PIT-CP procedure over CQR or even RCP is that it produces corrected conformal prediction regions for any miscoverage level $\alpha \in (0, 1)$ with linear $O(n)$ complexity, requiring only the computation of an empirical quantile $\hat{q}_{1 - \alpha}$.

\bigbreak

In the following section, we assume that $\mathcal{X} \subseteq \mathbb{R}^p$ and $\mathcal{Y} \subseteq \mathbb{R}^d$ for some $p, d \geq 1$. Also, let $\mathcal{D}_N' \coloneqq \left\{ (X_i',Y_i') \right\}_{i=1}^N$ be an auxiliary i.i.d. training sample independent of the calibration data $\mathcal{D}_n$, and define $S_i' \coloneqq s(X_i', Y_i')$ for all $i \in \llbracket N \rrbracket$. To ensure the existence of conditional densities, we further assume that for all $x \in \mathcal{X}$, the conditional distribution of $S \mid X = x$ is absolutely continuous with respect to the Lebesgue measure, thereby satisfying Assumption~\ref{ass:cont_cond}, and we let $\{ p(\cdot \mid x) : x \in \mathcal{X} \}$ denote the corresponding set of conditional densities.

\subsection{Minimax-Optimal Kernel-Based Estimators} \label{subsec:minimax}

Although such approaches may not be the most practical from a computational perspective, \cite{li2022minimax} detail a kernel-based procedure that achieves minimax-optimal rates for conditional density estimation under the expected total variation loss, over suitable smoothness classes. This result provides a useful theoretical benchmark for our setting. In this subsection, we also require the following boundedness assumption. For the sake of simplicity, we take a unit bound without loss of generality.

\begin{assumption}[Bounded Random Variables] \label{ass:bounded}
The random variables $X \in \mathbb{R}^p$ and $S \in \mathbb{R}$ satisfy $X \in [0, 1]^p$ and $S \in [0, 1]$ almost surely.
\end{assumption}

The difficulty of conditional density estimation is mostly governed by both the smoothness properties of the underlying distribution and the dimensions of the input and output spaces. In our framework, the input variable is $p$-dimensional, while the response variable is one-dimensional, since we focus on estimating the conditional density of the nonconformity score $S \coloneqq s(X, Y)$. As such, the estimation problem falls into a $(p, 1)$-dimensional conditional density estimation setting, which may be substantially easier than the estimation of the full conditional density of the target.

We now introduce two formal notions of smoothness for the conditional distribution. The first measures the smoothness of $P_{S \mid X = x}$ uniformly over the features $x \in [0, 1]^p$, whereas the latter quantifies the smoothness (\textit{i.e.}, the rate of change) of the mapping $x \mapsto P_{S \mid X = x}$. These definitions will be used to characterize the statistical complexity of the problem and to state minimax rates.

\begin{assumption}[H\"older Smooth Conditional Density] \label{ass:holder_smooth}
The conditional densities $\{ p(\cdot \mid x) : x \in [0, 1]^p \}$ are $\beta$-H\"older smooth, \textit{i.e.}, there exist $\beta > 0$ and $W_1 > 0$ such that for all $s, s' \in [0, 1]$ and for all $x \in [0, 1]^p$, the function $t \mapsto p(t \mid x)$ is $\ell \coloneqq \lfloor \beta \rfloor$ times differentiable and satisfies
\[
\left\vert \frac{\partial^{\ell}}{\partial s^{\ell}} p(s \mid x) - \frac{\partial^{\ell}}{\partial(s')^{\ell}} p(s' \mid x) \right\vert \leq W_1 \vert s - s' \vert^{\beta - \ell}.
\]
\end{assumption}

\begin{assumption}[Total Variation Smoothness] \label{ass:tv_smooth}
The conditional distributions $\{ P_{S \mid X = x} : x \in [0, 1]^p \}$ are $\gamma$-total variation smooth, \textit{i.e.}, there exist $0 < \gamma \leq 1$ and $W_2 > 0$ such that for all $(x, x') \in [0, 1]^p \times [0, 1]^p$,
\[
d_{TV}\left( P_{S \mid X = x}, P_{S \mid X = x'} \right) \leq W_2 \Vert x - x' \Vert_1^{\gamma}.
\]
\end{assumption}

\bigbreak

Under these assumptions, Theorem~3.5 (or 4.1 for the lower bound) of \cite{li2022minimax} can be invoked to derive minimax-optimal rates for conditional density estimation under the expected total variation distance. Combining this estimator with Corollary~\ref{cor:exp_prob_bounds} and Theorem~\ref{thm:sym_diff} yields the convergence guarantees stated in Theorem~\ref{thm:minimax_high_prob_bound}.

\begin{theorem}[Minimax Optimal Conditional Density Estimator] \label{thm:minimax_high_prob_bound}
Suppose Assumptions~\ref{ass:iid}, \ref{ass:cont_cond}, \ref{ass:bounded}, \ref{ass:holder_smooth}, and \ref{ass:tv_smooth} hold. Then, there exists a conditional density estimator satisfying $(\mathrm{i})$, and if Assumption~\ref{ass:increasing} also holds and ties in $\{ \widehat{S}_i \}_{i=1}^n$ occur with zero probability, then $(\mathrm{ii})$ is satisfied, where:
\begin{enumerate}[label=(\roman*), noitemsep]
    \item $\widehat{\Delta}(X) = O_P\left( N^{-\frac{1}{1/\beta + p/\gamma + 2}} \right)$.
    \item $\mathbb{P}\left( Y \in \widehat{C}_{1 - \alpha}(X) \symdif C^*_{1 - \alpha}(X) \right) = O_P\left( N^{-\frac{1}{1/\beta + p/\gamma + 2}} + \sqrt{\frac{\log(n)}{n}} \right)$.
\end{enumerate}
\end{theorem}

The proof of Theorem~\ref{thm:minimax_high_prob_bound} is provided in Appendix~\ref{pf:minimax_high_prob_bound}.

\bigbreak

While these upper bounds are achieved by a hybrid construction combining kernel and histogram estimators, and such estimators are generally not recommended in practice due to their poor empirical performance in many settings, they remain convenient and provide a clean insight into how the bound scales with the smoothness parameters $\beta$ and $\gamma$. More specifically, $\beta$ uniformly controls the smoothness of each conditional density in the score variable, whereas $\gamma$ controls how smoothly the conditional law varies with $x \in [0, 1]^p$. In particular, $\beta$ may be arbitrarily large, and larger values correspond to an easier conditional density estimation problem.

\subsection{Likelihood-Based Estimation and Forward Kullback--Leibler Divergence} \label{subsec:lik}

To estimate $P_{S \mid X}$, a common strategy is to use Maximum Likelihood Estimation within a parametric, but sufficiently rich, family $\{ p_\theta(\cdot \mid \cdot) : \theta \in \Theta \}$. When the model is sufficiently expressive, the Maximum Likelihood Estimator yields a model close to the true conditional distribution in this divergence. To see why this choice is especially natural in our framework, we relate the maximum likelihood procedure to the expected forward Kullback--Leibler divergence appearing in Corollary~\ref{cor:exp_prob_bounds}.

\bigbreak

Assume $\mathbb{E}\left[ \left\vert \log p(S \mid X) \right\vert \right] \vee \mathbb{E}\left[ \left\vert \log p_\theta(S \mid X) \right\vert \right] < +\infty$ for all $\theta \in \Theta$, and $P_{S \mid X = x} \ll P_{\theta \mid X = x}$. Letting
\[
\mathcal{L}(\theta) \coloneqq \mathbb{E}\left[ \log p_\theta(S \mid X) \right],
\]
the entropy-divergence decomposition is given by
\begin{align*}
\mathcal{L}(\theta)
&= \mathbb{E}\left[ \log p_\theta(S \mid X) \right] \\
&= \mathbb{E}\left[ \log p(S \mid X) \right] - \mathbb{E}\left[ \log \frac{p(S \mid X)}{p_\theta(S \mid X)} \right] \\
&= \mathbb{E}\left[ \log p(S \mid X) \right] - \mathbb{E}\left[ D_{KL}\left( P_{S \mid X} \parallel P_{\theta \mid X} \right) \right].
\end{align*}

Thus maximizing $\mathcal{L}(\theta)$ is equivalent to minimizing the expected forward Kullback--Leibler divergence
\[
\operatorname*{\arg\max}_{\theta \in \Theta} \mathcal{L}(\theta) = \operatorname*{\arg\min}_{\theta \in \Theta} \mathbb{E}\left[ D_{KL}\left(P_{S \mid X} \parallel P_{\theta \mid X}\right) \right].
\]

It is important to emphasize, however, that this correspondence is specific to the forward Kullback--Leibler divergence $\mathbb{E}\left[ D_{KL}\left( P_{S \mid X} \parallel P_{\theta \mid X} \right) \right]$. By contrast, the reverse divergence $\mathbb{E}\left[ D_{KL}\left( P_{\theta \mid X} \parallel P_{S \mid X} \right) \right]$ is in general not equivalent to the $\log$-likelihood objective, and may behave quite differently, but may be useful to consider nonetheless.

The empirical Maximum Likelihood Estimator is likewise the empirical risk minimization counterpart
\[
\hat{\theta} \in \operatorname*{\arg\max}_{\theta \in \Theta} \frac{1}{N}\sum_{i=1}^N \log p_\theta(S_i' \mid X_i').
\]

Hence, likelihood-based estimation directly aligns with our theoretical bounds. We now detail two state-of-the-art, MLE-based conditional density estimators, which can be readily used in our framework.

\subsection{Mixture Density Networks}

A historical choice to model the conditional distribution $P_{S \mid X = x}$ for all $x \in \mathbb{R}^p$ is to use a Mixture Density Network (MDN) \citep{bishop1994mixture}, defined by
\[
p_\theta(s \mid x) = \sum_{j=1}^m \pi_j(x; \theta) \mathcal{N}\left( s; \mu_j(x; \theta), \sigma_j^2(x; \theta) \right),
\]
with a fixed number of components $m \geq 1$, and class proportions in the unit simplex satisfying $\pi_j(x; \theta) \geq 0$ and $\sum_{j=1}^m \pi_j(x; \theta) = 1$.

Mixture Density Networks offer a flexible parametric framework that can capture multimodal and heteroscedastic conditional distributions, extending Gaussian Mixture Models to parametric mixtures indexed by $x \in \mathcal{X}$. When the number of components $m$ grows, MDNs become increasingly expressive and can approximate a wider class of conditional densities, potentially arbitrarily well under suitable regularity assumptions on the target distribution and sufficient capacity of the neural network parametrization. For instance, under fairly restrictive assumptions, one may also obtain explicit convergence rates in terms of the expected forward Kullback--Leibler divergence \citep{mendes2011convergence}.

In our setting, a noteworthy advantage of MDNs is that the induced conditional CDF is available in closed form, without resorting to numerical integration or Monte Carlo sampling. Indeed, given a parameter estimate $\hat{\theta}$, and letting $\Phi$ be the CDF of a standard normal distribution, one obtains
\[
\widehat{F}_{S \mid X = x}(s) = \sum_{j=1}^m \pi_j(x; \hat{\theta}) \Phi\left( \frac{s - \mu_j(x; \hat{\theta})}{\sigma_j(x; \hat{\theta})} \right).
\]

The plug-in corrected score is therefore
\[
\hat{s}(x, y) = \widehat{F}_{S \mid X = x}\left(s(x, y)\right),
\]
which can be evaluated directly without numerical approximation. The inverse can also be approximated through numerical inversion.

\bigbreak

We thus obtain the PIT-CP-MDN procedure summarized in Algorithm~\ref{alg:pit_cp_mdn}. As previously noted, as our procedure uses the split conformal prediction framework, it only requires storing the transformed scores $\{ U_i \}_{i=1}^n$. Because these scores are constructed independently of the target miscoverage rate $\alpha \in (0, 1)$, it is straightforward to generate predictions for multiple levels simultaneously without retraining.

\begin{algorithm}[H] 
\caption{PIT-CP-MDN}
\label{alg:pit_cp_mdn}
\begin{algorithmic}[1] 
\Require Training data $\mathcal{D}_N' = \left\{ (X_i', Y_i') \right\}_{i=1}^N$, calibration data $\mathcal{D}_n = \left\{ (X_i, Y_i) \right\}_{i=1}^n$, base score $s : \mathcal{X} \times \mathcal{Y} \to \mathbb{R}$, parametric MDN conditional densities $\{ p_\theta(\cdot \mid \cdot) : \theta \in \Theta \}$ with $m$ components, miscoverage level $\alpha \in (0, 1)$, test point $x \in \mathcal{X}$

\Statex
\Statex \textbf{Compute scores}
\State Compute $S_i' \coloneqq s(X_i', Y_i')$ for $i=1, \dots, N$, and $S_i \coloneqq s(X_i, Y_i)$ for $i=1, \dots, n$

\Statex
\Statex \textbf{Fit}
\State Compute or approximate the MLE $\hat{\theta} \in \operatorname*{\arg\max}_{\theta \in \Theta} \sum_{i=1}^N \log p_\theta(S_i' \mid X_i')$

\Statex
\Statex \textbf{Calibrate}
\State Compute $U_i \coloneqq \widehat{F}_{S \mid X = X_i}(S_i)$ for $i=1,\dots, n$
\State $\hat{q}_{1 - \alpha} \coloneqq$ the $\left\lceil (n + 1)(1 - \alpha) \right\rceil$-th order statistic of $\{ U_i \}_{i=1}^n \cup \{ +\infty \}$

\Statex
\State \Return $\widehat{C}_{1 - \alpha}(x) \coloneqq \left\{ y \in \mathcal{Y} : \widehat{F}_{S \mid X = x}\left( s(x, y) \right) \leq \hat{q}_{1 - \alpha} \right\}$
\end{algorithmic}
\end{algorithm}

\subsection{Conditional Normalizing Flows}

A second, more modern class of models is provided by Conditional Normalizing Flows (CNFs) \citep{papamakarios2021normalizing}. Conditional Normalizing Flows provide a highly flexible alternative to finite mixtures, since they do not impose any explicit parametric shape on the conditional distribution beyond the existence of a differentiable monotone transport to a simple base distribution. The strong appeal of Conditional Normalizing Flows has motivated the development of many structured architectures such as Continuous Normalizing Flows \citep{chen2018neural}, Neural Spline Flows \citep{durkan2019neural}, Masked Autoregressive Flows \citep{papamakarios2017masked}, Sum-of-Squares Polynomial Flows \citep{jaini2019sum}, RealNVP \citep{dinh2016density}, and Bernstein Polynomial Flows \citep{arpogaus2021probabilistic, arpogaus2023short}, each with respective strengths and drawbacks. 

In the general multivariate setting, additional structural choices also become important, and triangular transports such as Knothe--Rosenblatt transformations \citep{rosenblatt1952remarks} play a central role in simplifying the computation of the $\log$-determinant of the flow's Jacobian. In our framework, however, this distinction is irrelevant because the target variable is one-dimensional.

Indeed, in our case the flow only needs to learn, for each $x \in \mathbb{R}^p$, a strictly increasing differentiable transport $f_\theta(\cdot \mid x) : \mathbb{R} \to \mathbb{R}$, mapping $s \in \mathbb{R}$ to $z \in \mathbb{R}$ through $z \coloneqq f_\theta(s \mid x)$. Equipping the codomain with a strictly positive base probability density $p_z$, and by a change of variables, the model induces the conditional density
\[
p_\theta(s \mid x) = p_z\left( f_\theta(s \mid x) \right) \left\vert \frac{\partial}{\partial s} f_\theta(s \mid x) \right\vert.
\]

\bigbreak

In addition, to further simplify the procedure and avoid explicitly mapping $z \in \mathbb{R}$ into the probability space $[0, 1]$, we state Lemma~\ref{lem:circ_invariance}, which shows that conformal prediction regions remain invariant under any fixed strictly increasing transformation.

\begin{lemma}[Strictly Increasing Composition Invariance] \label{lem:circ_invariance}
The split conformal prediction procedure described in Subsection~\ref{subsec:setting} is invariant under composing the nonconformity score function $s$ by any strictly increasing function $g : \mathbb{R} \to \mathbb{R}$, yielding $\tilde{s} = g \circ s$.
\end{lemma}

The proof of Lemma~\ref{lem:circ_invariance} is provided in Appendix~\ref{pf:circ_invariance}.

\bigbreak

Through Lemma~\ref{lem:circ_invariance}, we obtain the PIT-CP-CNF procedure summarized below in Algorithm~\ref{alg:pit_cp_cnf}. The key idea is that conformal prediction sets are invariant under strictly increasing transformations of the score, so that the calibration can be performed directly in the latent $z$-space without loss of validity. This yields numerically stable computations and avoids an unnecessary mapping to $[0, 1]$. The same observations regarding score storage and the lack of a need for retraining apply here as well.

Indeed, for each $i \in \llbracket n \rrbracket$, we define the transformed calibration variables
\[
U_i \coloneqq F_z^{-1}\left( F_{\hat{\theta}}\left( S_i \mid X_i \right) \right),
\]
where $F_{\hat{\theta}}(\cdot \mid x)$ denotes the conditional cumulative distribution function associated with the learned density $p_{\hat{\theta}}(\cdot \mid x)$, and $F_z^{-1}$ is the inverse CDF of the base density $p_z$. By construction, $F_z^{-1}$ is continuous and strictly increasing, because $p_z$ admits a continuous density.

Moreover, by a change of variables, and since by assumption $f_{\hat{\theta}}$ is strictly increasing in its first argument, we may also equivalently write
\[
F_{\hat{\theta}}(s \mid x) = F_z\left( f_{\hat{\theta}}(s \mid x) \right) \implies F_z^{-1}\left( F_{\hat{\theta}}(s \mid x) \right) = f_{\hat{\theta}}(s \mid x).
\]

This identity allows the transformation to be applied directly through $f_{\hat{\theta}}$, avoiding the need to explicitly compute the intermediate probability integral transform in $[0, 1]$. Consequently, the inverse mapping can also be readily obtained through $f_{\hat{\theta}}^{-1}$.

\begin{algorithm}[H]
\caption{PIT-CP-CNF}
\label{alg:pit_cp_cnf}
\begin{algorithmic}[1]
\Require Training data $\mathcal{D}_N' = \left\{ (X_i', Y_i') \right\}_{i=1}^N$, calibration data $\mathcal{D}_n = \left\{ (X_i, Y_i) \right\}_{i=1}^n$, base score $s : \mathcal{X} \times \mathcal{Y} \to \mathbb{R}$, parametric CNF $f_\theta(\cdot \mid \cdot)$ and conditional densities $\{ p_\theta(\cdot \mid \cdot) : \theta \in \Theta \}$, miscoverage level $\alpha \in (0, 1)$, test point $x \in \mathcal{X}$
\Statex
\Statex \textbf{Compute scores}
\State Compute $S_i' \coloneqq s(X_i',Y_i')$ for $i=1,\dots,N$, and $S_i \coloneqq s(X_i,Y_i)$ for $i=1,\dots,n$
\Statex
\Statex \textbf{Fit}
\State Compute or approximate the MLE $\hat{\theta} \in \operatorname*{\arg\max}_{\theta \in \Theta} \sum_{i=1}^N \log p_\theta(S_i' \mid X_i')$
\Statex
\Statex \textbf{Calibrate}
\State Compute $U_i \coloneqq f_{\hat{\theta}}(S_i \mid X_i)$ for $i=1,\dots,n$
\State $\hat{q}_{1 - \alpha} \coloneqq$ the $\left\lceil (n + 1)(1 - \alpha) \right\rceil$-th order statistic of $\{ U_i \}_{i=1}^n \cup \{+\infty\}$
\Statex
\State \Return $\widehat{C}_{1 - \alpha}(x) \coloneqq \left\{ y \in \mathcal{Y} : f_{\hat{\theta}}\left( s(x, y) \mid x \right) \leq \hat{q}_{1 - \alpha} \right\}$
\end{algorithmic}
\end{algorithm}

\section{Numerical Experiments} \label{sec:num}

In this section, we investigate the empirical performance of our method through both controlled simulation studies and real-world data experiments. The simulation setting allows us to precisely control the data-generating assumptions and illustrate the invariance of our procedure with respect to the choice of the base nonconformity score. In addition, we consider a real-world regression task based on the SARCOS dataset \citep{vijayakumar2000sarcos}, where the objective is to assess performance under more complex and less structured conditions. Across all experiments, we compare our approach against a range of alternative methods, depending on their applicability in each setting.

\bigbreak

All simulations and figures reported in this section were produced using our companion \href{https://pypi.org/project/pitcp/}{\texttt{pitcp}} Python package \underline{version 0.10.0}, available on PyPI, which we developed to implement the two algorithms presented in Section~\ref{sec:algorithms}. This package is fully \texttt{scikit-learn}-compatible, and built upon the \href{https://github.com/probabilists/zuko}{\texttt{zuko}} framework \citep{rozet2022zuko}. Reproducible scripts used to generate the experiments are available in its parent \href{https://github.com/felixlaplante0/pitcp}{GitHub} repository \underline{under the \textit{paper} branch}, with minor numerical differences expected across platforms or hardware configurations, due to the nondeterministic behavior of some \texttt{PyTorch} operations.

\subsection{Simulation Study}

\subsubsection{Toy Example} \label{par:toy}

We first illustrate the PIT-CP procedure using a simple synthetic heteroscedastic regression setting inspired by \cite{guan2023localized}, but with smoother distribution changes with respect to the features in order to reflect Assumption~\ref{ass:tv_smooth}. The design is as follows. We consider one-dimensional features $X_i$ uniformly distributed on $[-1, 1]$ and response $Y_i$ generated conditionally on $X_i$ as
\[ 
Y_i \mid X_i = x \overset{\mathrm{i.i.d.}}{\sim} \mathcal{N}\left( 0, \sigma(x)^2 \right), \quad \sigma(x) = \vert 1 - 2x^2 \vert + \frac{1}{10}.
\]

This design gives a conditional distribution that exhibits strong heteroscedasticity as there is more than a $10$-fold increase in noise scale. It also provides the advantage of being easily visualizable, taking a candy shape as seen in Figure~\ref{fig:guan_simulation}. Although many different methods, including well-specified CQR, are capable of achieving conditional coverage in this toy example, we demonstrate that our PIT-CP procedure successfully achieves approximate conditional coverage regardless of the choice of the base nonconformity score and target miscoverage level $\alpha$, while still preserving the marginal coverage guarantee inherent to split conformal prediction.

\bigbreak

We train a Sum-of-Squares Polynomial Flow-based correction, following Algorithm~\ref{alg:pit_cp_cnf}, on $N = 5000$ training samples of three distinct base scores, each calibrated on $n = 1000$ data points at a different target miscoverage level $\alpha$. We report an approximate training time of $20$ seconds on an \textit{AMD Ryzen 5 5600} six-core processor, given $200$ epochs, a batch size of $512$ samples, and using the Adam optimizer with a learning rate of $10^{-3}$ \citep{kingma2014adam}. The selected scores and miscoverage levels are the absolute residual $s(x, y) = \vert y \vert$ for $\alpha = 0.3$, the oracle negative likelihood $s(x, y) = -\log p(y \mid x)$ for $\alpha = 0.2$, which corresponds to the Highest Predictive Density (oracle HPD-split) scenario, and the raw response $s(x, y) = y$ for $\alpha = 0.1$, yielding unilateral conformal intervals. Given that the data-generating process is symmetric around $y = 0$, we expect the PIT-corrected scores for the absolute residual (distance to the mean) and the HPD to coincide, provided the PIT-CP procedure is properly trained.

For each method, score, and miscoverage level $\alpha \in (0, 1)$, we then evaluate the Mean Absolute Error (MAE) between the achieved conditional coverage and the marginal coverage, which (up to finite-sample correction terms) is uniformly bounded across $\alpha \in (0, 1)$ by the $L^1$ conditional coverage gap $\mathbb{E}[ \widehat{\Delta}(X) ]$ that is theoretically controlled by the Kolmogorov--Smirnov distance in Equation~\eqref{eq:l1_l2_gap} of Corollary~\ref{cor:exp_prob_bounds}. We also compare to standard CQR with boosting-based quantile regressors using the CatBoost framework \citep{prokhorenkova2018catboost}.

Formally, considering the conditional coverage achieved by some procedure yielding a conformal prediction region $\widehat{C}_{1 - \alpha}(x)$ at a point $x \in [-1, 1]$ given a miscoverage level $\alpha \in [0, 1]$
\[
\hat{c}_{1 - \alpha}(x) \coloneqq \mathbb{P}\left( Y_{n + 1} \in \widehat{C}_{1 - \alpha}(X_{n + 1}) \mid \mathcal{D}_n, X_{n + 1} = x \right),
\]
the Mean Absolute Error can be obtained through
\[
\bar{c}_{1 - \alpha} = \frac{1}{2} \int_{-1}^1 \hat{c}_{1 - \alpha}(x) \, dx, \quad \mathrm{MAE} = \frac{1}{2} \int_{-1}^1 \left\vert \hat{c}_{1 - \alpha}(x) - \bar{c}_{1 - \alpha} \right\vert \, dx.
\]

\bigbreak

As illustrated in Figure~\ref{fig:guan_simulation}, the initial conformal prediction regions (blue) vary drastically in shape and conditional coverage depending on the chosen score. Applying the PIT correction (green) reshapes these regions to match the conditional data distribution. The corrected conditional coverage curves tightly fluctuate around each specified target coverage $1 - \alpha$ (black) and drastically reduce the Mean Absolute Error. This confirms that the PIT-based correction effectively mitigates the conditional coverage fluctuations induced by the original nonconformity scores, even surpassing CQR in all cases, although this could be solely explained by the choice of the quantile regressors. This empirical experiment demonstrates that the PIT-CP framework is fully agnostic to the base score, retains its geometry, maintains marginal coverage over $X$, and lastly does not need separate models for different quantiles. By contrast, while CQR also provides approximate conditional validity, it remains less flexible, particularly in multidimensional settings: it is restricted to bilateral conformal prediction regions and requires retraining to generalize across different miscoverage levels.

\begin{figure}[H]
    \centering
    \begin{subfigure}{0.32\textwidth}
        \centering
        \includegraphics[width=1.02\textwidth, trim=7 7 7 7, clip]{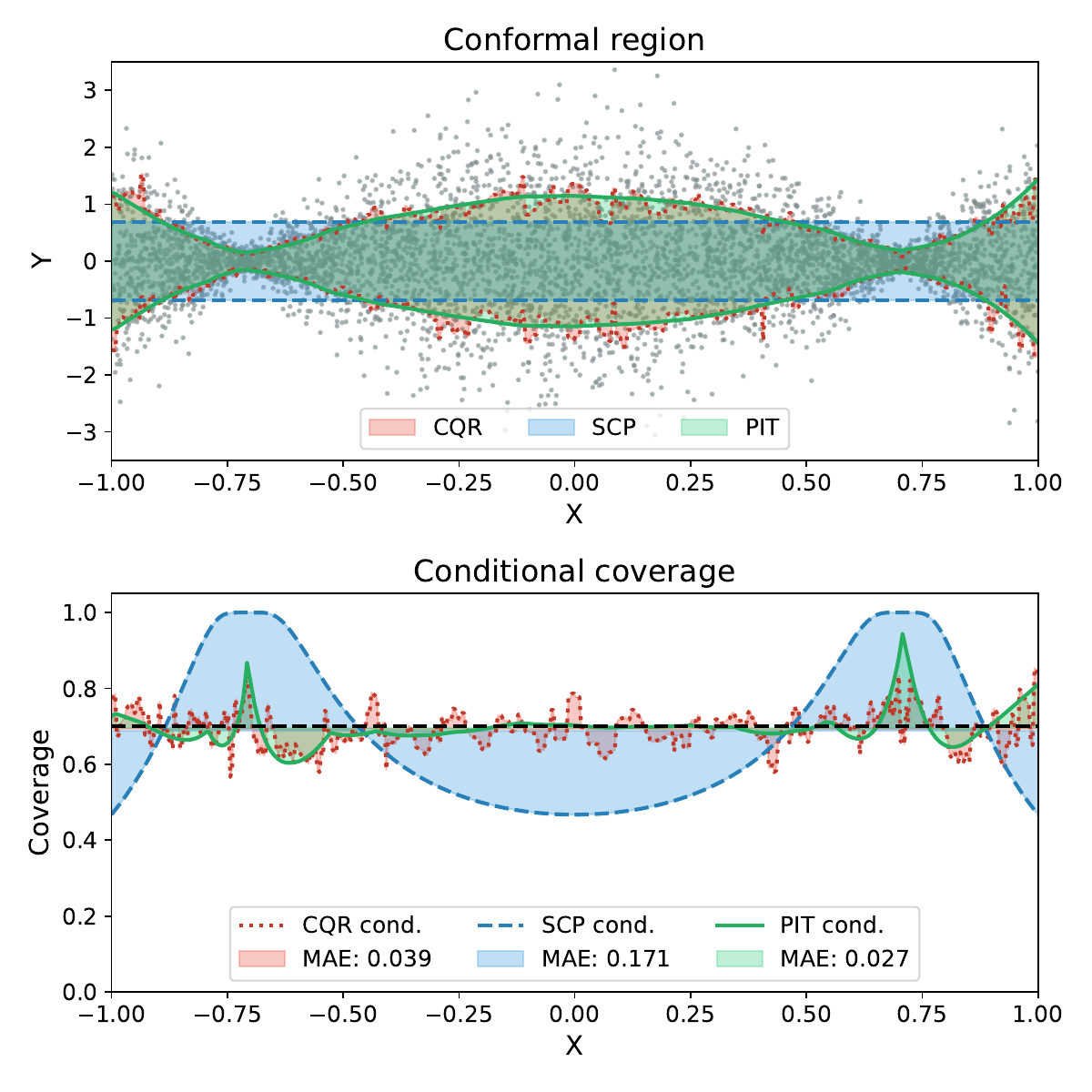}
        \caption{$s(x, y) = \vert y \vert, \, \alpha = 0.3$}
    \end{subfigure}
    \begin{subfigure}{0.32\textwidth}
        \centering
        \includegraphics[width=1.02\textwidth, trim=7 7 7 7, clip]{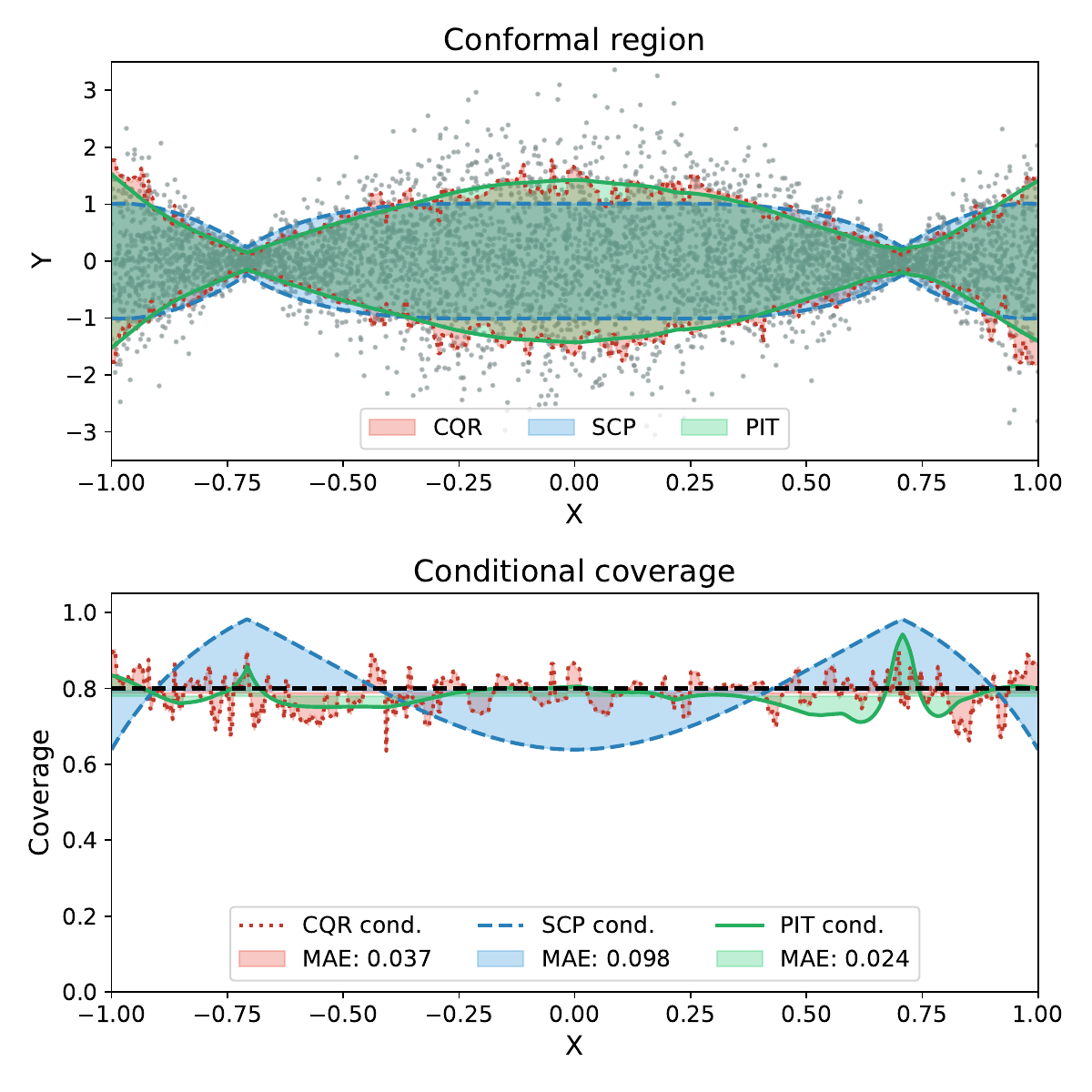}
        \caption{$s(x, y) = -\log p(y \mid x), \, \alpha = 0.2$}
    \end{subfigure}
    \begin{subfigure}{0.32\textwidth}
        \centering
        \includegraphics[width=1.02\textwidth, trim=7 7 7 7, clip]{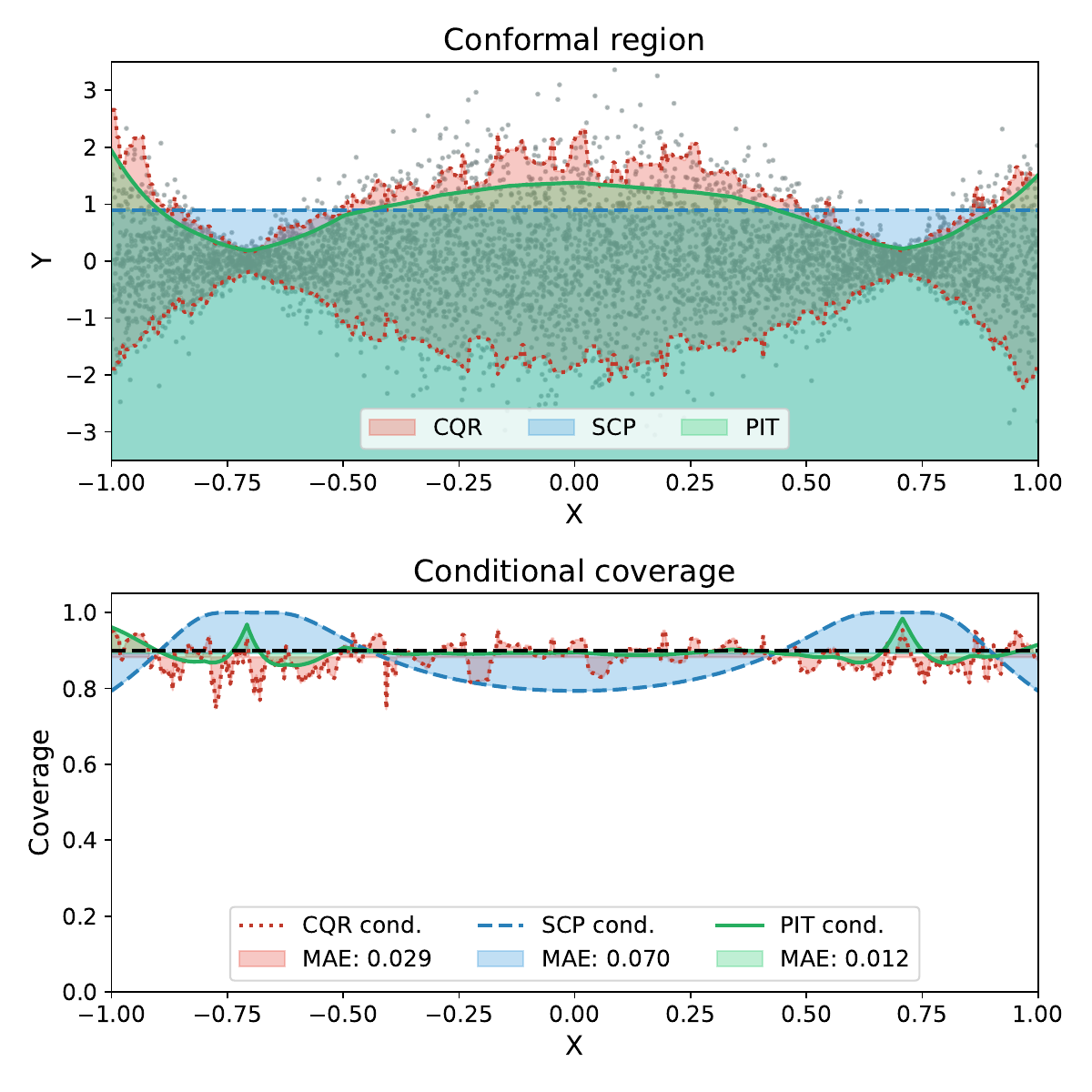}
        \caption{$s(x, y) = y, \, \alpha = 0.1$}
    \end{subfigure}
    \caption{Comparison between base and PIT-corrected conformal prediction regions evaluated across three nonconformity scores and target coverages (black dashed lines). The PIT correction reliably restores approximate conditional coverage and maintains marginal validity as well as geometric shape, independent of the base score. A total of $n_\mathrm{test} = 5000$ points are scattered.}
    \label{fig:guan_simulation}
\end{figure}

\subsubsection{Convergence Study}

We now turn to the empirical convergence of our PIT-CP procedure as the number of training samples $N$ increases. We consider the same data-generating process as in the previous Paragraph~\ref{par:toy}.

Since the distribution is centered, we apply the PIT-CP procedure to the absolute residual base score $s(x, y) = \vert y \vert$ calibrated on $n = 1000$ samples. We evaluate the conditional coverage on a uniform grid of $98$ miscoverage levels $\alpha \in \{ \frac{k}{99} \}_{k=1}^{98}$. The conditional coverage is then compared to the marginal coverage through the $L^1$ conditional coverage gap, namely $\mathbb{E}\left[ \widehat{\Delta}(X) \right]$. 

Using the same notation as before, the evaluation metric for a fixed conformal prediction region function is approximated, assuming $n$ sufficiently large, by
\[
\mathbb{E}\left[ \widehat{\Delta}(X) \right] \approx \mathbb{E}\left[ \widehat{\Delta}(X) \mid \mathcal{D}_n \right] \approx \frac{1}{2} \int_{-1}^1 \max_{\alpha \in \{ \frac{k}{99} \}_{k=1}^{98}} \left\vert \hat{c}_{1 - \alpha}(x) - \bar{c}_{1 - \alpha} \right\vert \, dx.
\]

Note that the prediction at various quantiles entails no time penalty as a single probabilistic model encompasses the whole conditional distribution of the score at once, which is one of the most appealing aspects of the PIT-CP procedure. The experiment is then repeated for $n_\mathrm{runs} = 10$ independent runs for varying training sample sizes $N \in \{ 0, 1000, 2000, 3000, 4000, 5000 \}$. Here, $N = 0$ corresponds to the baseline performance without any post-processing correction. We compare the performance of two density estimators presented in Section~\ref{sec:algorithms}: a Gaussian Mixture Model (GMM, representing Algorithm~\ref{alg:pit_cp_mdn}) parameterized with $m = 5$ components (which empirically proved to be a reasonable compromise between expressiveness and optimization stability), and a Sum-of-Squares Polynomial Flow (SOSPF, representing Algorithm~\ref{alg:pit_cp_cnf}). The optimization is carried out following the same procedure as previously described in Paragraph~\ref{par:toy}.

\bigbreak

As depicted in Figure~\ref{fig:convergence}, both methods substantially reduce the $L^1$ error between the exact conditional coverage and the marginal coverage as the number of training samples grows, thus significantly improving conditional validity. While an elbow at $\approx 1000$ samples can be noted, indicating diminishing returns, the $L^1$ conditional gap still continues to improve well beyond this point without encountering a clear bottleneck. The SOSPF model exhibits a clear advantage in this setting with a $7$-fold reduction compared to the uncorrected baseline, which we attribute to its superior flexibility and better optimization quality. This suggests preferring conditional normalizing flows over mixture density networks. However, alternative architectures might yield significantly stronger results but were not explored in this work. Practitioners should therefore compare and tune models to maximize performance for each specific application, for instance using convolution-based networks for specialized vision tasks.

\begin{figure}[H]
    \centering
    \includegraphics[width=0.6\textwidth, trim=7 7 7 7, clip]{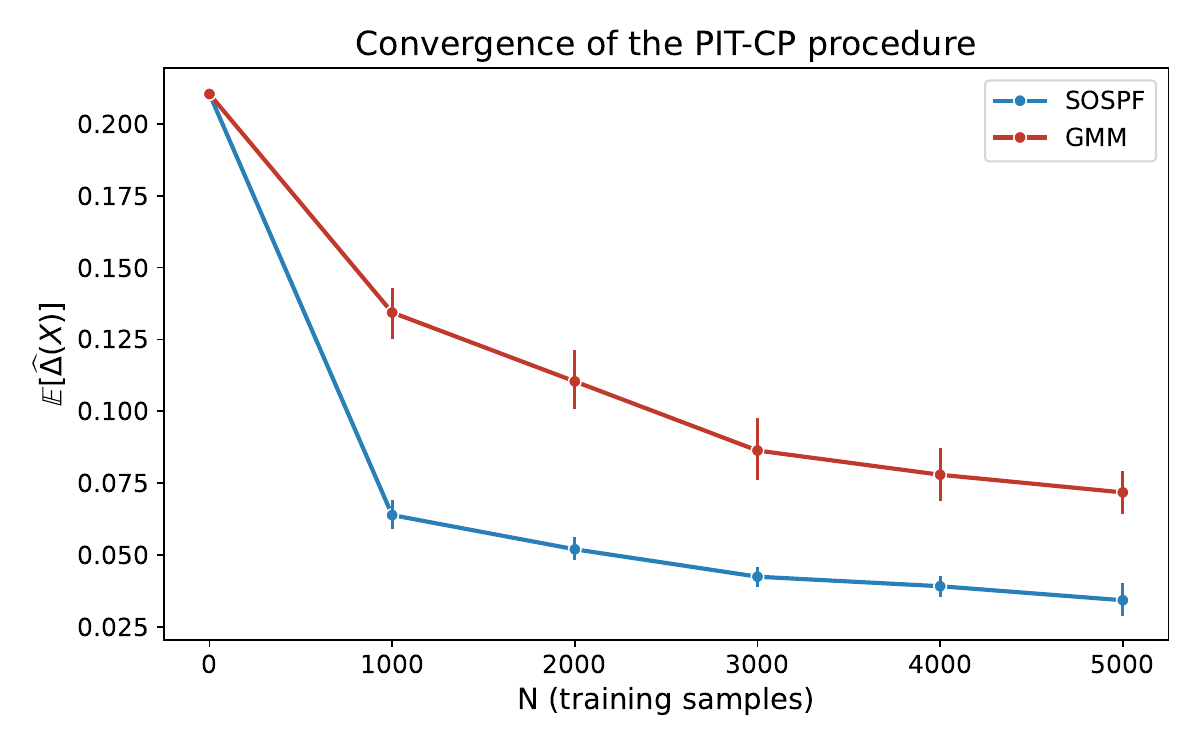}
    \caption{$L^1$ conditional coverage gap as a function of the number of training samples $N$ for GMM and SOSPF estimators. Error bars indicate the standard deviation computed on $n_\mathrm{test} = 5000$ points, over $n_\mathrm{runs} = 10$ repetitions.}
    \label{fig:convergence}
\end{figure}

\subsection{Real-World Data}

We evaluate the proposed methodology on the SARCOS dataset \citep{vijayakumar2000sarcos}, which provides a standard benchmark for multidimensional regression tasks involving robotic arm inverse dynamics. In our case, we aim to predict the joint torques of the $7$ degrees of freedom of a SARCOS anthropomorphic robot arm. The dataset consists of $44484$ total samples. We randomly shuffle the data and extract $4$ distinct parts: $1/3$ for training, $1/3$ for validation, $1/6$ for calibration, and $1/6$ for testing. The feature variable $X \in \mathcal{X}$ represents the input space with dimension $p = 21$, which includes position, velocity, and acceleration information. The outcome variable $Y$ corresponds to the $7$ joint torques, hence the outcome space is a seven-dimensional subset $\mathcal{Y} \subseteq \mathbb{R}^7$. Following the standard practice, both features and targets are normalized prior to model training, and the transformation is inverted during post-processing. Although the initial data were recorded as a time series, we rely on the assumption that the underlying data-generating process is Markovian, \textit{i.e.}, the torques are only affected by the current forces exerted on the arm and are independent of the past trajectory, ensuring independence of the observations.

\subsubsection{Base Score}

To quantify the accuracy of a prediction $\hat{y}$ with respect to a ground-truth measurement $y$, we rely on the $L^\infty$-norm of the standardized residuals. More precisely, as a natural base nonconformity score required by our PIT-CP methodology, we use, given features $x \in \mathcal{X}$
\[
s(x, y) = \left\Vert D^{-1} (y - \hat{f}(x)) \right\Vert_\infty,
\]
where $\hat{f}$ denotes a black-box predictor, and $D$ is a diagonal scaling matrix containing the empirical standard deviations of the residuals on the validation set. More precisely, the base point predictions $\hat{f}(x)$ are obtained using TabPFN \citep{grinsztajn2025tabpfn, hollmann2025tabpfn}, a freely available cutting-edge tabular foundation neural model trained on large amounts of synthetic data which requires no hyperparameter tuning. It is trained on the first split, comprising one third of the dataset, and consistently achieves $R^2$ scores exceeding $99\%$ across all targets. This nonconformity score yields conformal prediction regions consisting of hyperrectangles. Such regions are easily interpretable, being the Cartesian product of univariate intervals, and also equally penalize errors across any dimension.

\bigbreak

To summarize, evaluating this high-dimensional task using the TabPFN predictor and the $L^\infty$ nonconformity score provides a compelling application for our framework for at least three reasons:
\begin{enumerate}[itemsep=0pt]
    \item Coordinate-wise CQR requires a distinct model for each target coverage level $1 - \alpha$ and each dimension, leading to potentially significant training overhead. To ensure simultaneous conditional validity, it also relies on the assumption of conditional independence among the responses given $X$.
    \item Modeling the full conditional distribution in $\mathbb{R}^7$ via HPD-split is statistically demanding and it might be too optimistic to hope to achieve superior results in terms of volume compared to a well-established, state-of-the-art pre-trained predictive model.
    \item Identifying the highest density region directly in $\mathcal{Y}$ may produce sets that lack structural interpretation. Our approach preserves the original hyperrectangular geometry, while still approaching conditional validity.
\end{enumerate}

\subsubsection{Evaluation Metrics}

To empirically evaluate conditional validity, we perform K-means clustering ($K=10$) on the (rescaled) testing features to partition the input space into distinct bins. Given $\alpha \in (0, 1)$ and conformal prediction region function $\widehat{C}_{1 - \alpha}$, the conditional coverage gap is then computed as the peak-to-peak discrepancy in empirical coverage across these clusters $\{ C_k \}_{k=1}^K$. Formally, noting $\left\{ (X_i'', Y_i'') \right\}_{i=1}^{n_\mathrm{test}}$ the testing sample, we compute the gap as
\[
\mathrm{Gap}_{1 - \alpha} \coloneqq \max_{1 \leq k \leq K} \frac{1}{\vert C_k \vert} \sum_{i \in C_k} \mathbb{I}\left\{ Y_i'' \in \widehat{C}_{1 - \alpha}(X_i'') \right\} - \min_{1 \leq k \leq K} \frac{1}{\vert C_k \vert} \sum_{i \in C_k} \mathbb{I}\left\{ Y_i'' \in \widehat{C}_{1 - \alpha}(X_i'') \right\}.
\]

This metric is preferred over the standard worst-slab coverage \citep{cauchois2021knowing} because it penalizes both over- and under-coverage. Other related clustering-based evaluation metrics closer to the $L^1$ error have also been proposed, for instance by \cite{ding2023class}. Moreover, in our setting, assuming the clusters are roughly of equal size, the coverage is computed using approximately $750$ data points, which provides a robust estimate. We also report the first quartile, median, and third quartile of the conformal prediction regions' volume. We omit the average volume as it scales as $\ell^7$, where $\ell$ denotes the side length of a prediction region, and is therefore too sensitive to outliers. Details on the computation of the volumes for the HPD-split and CONTRA algorithms are provided in Appendix~\ref{subsec:hpd} and \ref{subsec:contra}, respectively.

\subsubsection{Results}

We compare our proposed PIT-CP against several baselines: SCP (with base score $s$), Coordinate-wise CQR calibrated using the $L^\infty$-norm of the component-wise nonconformity scores, HPD-split, and CONTRA. Our PIT-CP relies on a Sum-of-Squares Polynomial Flow, HPD-split and CONTRA share the same Masked Autoregressive Flow architecture, while CQR is built upon CatBoost quantile regressors to form hyperrectangles. With the exception of the PIT-CP procedure, all models are trained on the combined training and validation splits, as PIT-CP relies exclusively on the validation split to avoid data leakage from the TabPFN training phase, with no hyperparameter tuning in all cases. The training of the neural parametric models is achieved with a batch size of $1024$ samples, using the Adam optimizer with a learning rate of $10^{-3}$ for over $20 000$ optimization steps, taking approximately $2$ minutes per model. The experiment is repeated for $n_\mathrm{runs} = 10$ independent runs, where the first third of the data used by the TabPFN predictor remains unshuffled while the remainder is randomly shuffled to ensure the statistical significance of our results.

\bigbreak

As shown in Table~\ref{tab:real_world} below, all methods remain approximately conditionally valid, meaning their coverage gaps are much smaller than that of SCP. However, density-based methods (HPD, CONTRA) and our PIT-CP approach achieve the lowest conditional coverage gaps overall with no statistically significant difference between them, and clearly outperform SCP and CQR. In hindsight, this similarity is expected as all these procedures depend, whether implicitly or explicitly, on distributional invariance to achieve approximate conditional validity. In particular, our PIT-CP procedure achieves approximate conditional validity while retaining the interpretable hyperrectangular geometry induced by the $L^\infty$ base score, and also producing prediction regions vastly tighter in volume than CQR, HPD-split, and CONTRA. Counterintuitively, minimizing the volume of conformal prediction regions in a fully distributional setting requires prediction scores that are not only well-calibrated on average over the entire conditional distribution, but also sufficiently tight. In our experiments, however, even a simple MLP-based regression score was unable to produce prediction sets more efficient than those obtained using the base score, suggesting that conditional validity in fully distributional approaches such as HPD-split does not necessarily imply good prediction volumes. We do not claim, however, that these methods could not become more efficient through optimization of the neural architectures, hyperparameter tuning, or selection of more suitable normalizing flows. Nonetheless, these improvements would require additional development time (which was not available to us) and careful empirical optimization which are not needed by our post-processing correction. 

\begin{table}[H]
    \centering
    \renewcommand{\arraystretch}{0.9}
    \begin{tabular}{clcccc}
        \toprule
        $1 - \alpha$ &Method &Gap &Vol. $1^\mathrm{st}$ quartile &Vol. Median &Vol. $3^\mathrm{rd}$ quartile \\
        \midrule
        \multirow{5}{*}{0.6} 
        &SCP &$0.53 \pm 0.03$& $\underline{3.64 \pm 0.23}$& $\underline{3.64 \pm 0.23}$& $\boldsymbol{3.64 \pm 0.23}$ \\
        &CQR &$\underline{0.27 \pm 0.02}$& $200.81 \pm 10.78$& $408.85 \pm 19.18$& $1009.65 \pm 48.90$ \\
        &HPD-split &$\boldsymbol{0.10 \pm 0.03}$ &$14.49 \pm 1.67$ &$42.54 \pm 3.88$ &$172.02 \pm 22.14$ \\
        &CONTRA &$\boldsymbol{0.10 \pm 0.03}$ &$14.67 \pm 1.72$ &$43.02 \pm 3.68$ &$176.23 \pm 22.89$ \\
        &PIT-CP &$\boldsymbol{0.09 \pm 0.04}$& $\boldsymbol{0.28 \pm 0.02}$& $\boldsymbol{2.88 \pm 0.15}$& $\underline{27.23 \pm 3.06}$ \\
        \midrule
        \multirow{5}{*}{0.7}
        &SCP &$0.47 \pm 0.03$& $\underline{12.91 \pm 0.80}$& $\underline{12.91 \pm 0.80}$& $\boldsymbol{12.91 \pm 0.80}$ \\
        &CQR &$\underline{0.24 \pm 0.02}$& $421.18 \pm 21.50$& $936.87 \pm 38.98$& $2478.09 \pm 101.30$ \\
        &HPD-split &$\boldsymbol{0.09 \pm 0.02}$ &$27.37 \pm 3.11$ &$80.28 \pm 6.75$ &$325.26 \pm 39.64$ \\
        &CONTRA &$\boldsymbol{0.09 \pm 0.02}$ &$27.82 \pm 3.49$ &$81.74 \pm 7.71$ &$335.21 \pm 47.20$ \\
        &PIT-CP &$\boldsymbol{0.08 \pm 0.02}$& $\boldsymbol{0.64 \pm 0.06}$& $\boldsymbol{6.86 \pm 0.36}$& $\underline{66.80 \pm 6.56}$ \\
        \midrule
        \multirow{5}{*}{0.8}
        &SCP &$0.40 \pm 0.04$& $\underline{53.26 \pm 3.53}$& $\underline{53.26 \pm 3.53}$& $\boldsymbol{53.26 \pm 3.53}$ \\
        &CQR &$\underline{0.19 \pm 0.03}$& $1158.22 \pm 49.26$& $2704.44 \pm 109.33$& $7339.42 \pm 318.24$ \\
        &HPD-split &$\boldsymbol{0.08 \pm 0.02}$ &$58.57 \pm 7.40$ &$171.76 \pm 15.40$ &$694.46 \pm 76.89$ \\
        &CONTRA &$\boldsymbol{0.08 \pm 0.02}$ &$59.80 \pm 8.13$ &$175.20 \pm 15.68$ &$719.90 \pm 89.24$ \\
        &PIT-CP &$\boldsymbol{0.08 \pm 0.02}$& $\boldsymbol{1.76 \pm 0.19}$& $\boldsymbol{19.47 \pm 0.98}$& $\underline{195.73 \pm 17.90}$ \\
        \midrule
        \multirow{5}{*}{0.9}
        &SCP &$0.25 \pm 0.03$& $390.53 \pm 27.32$& $\underline{390.53 \pm 27.32}$& $\boldsymbol{390.53 \pm 27.32}$ \\
        &CQR &$\underline{0.12 \pm 0.02}$& $6830.21 \pm 428.09$& $15989.44 \pm 891.05$& $40938.43 \pm 2037.48$ \\
        &HPD-split &$\boldsymbol{0.06 \pm 0.02}$ &$\underline{171.87 \pm 26.41}$ &$503.58 \pm 52.89$ &$2056.48 \pm 223.79$ \\
        &CONTRA &$\boldsymbol{0.06 \pm 0.02}$ &$175.74 \pm 28.06$ &$515.41 \pm 53.74$ &$2126.31 \pm 295.68$ \\
        &PIT-CP &$\boldsymbol{0.06 \pm 0.01}$& $\boldsymbol{7.67 \pm 0.76}$& $\boldsymbol{88.67 \pm 5.41}$& $\underline{899.05 \pm 70.19}$ \\
        \bottomrule
    \end{tabular}
    \caption{Conditional coverage gap and region volumes across selected quantiles, including standard errors. Top-performing results are highlighted in \textbf{bold}, while second-best results are \underline{underlined} (indicating a separation $> 1 \, \mathrm{sd}$).}
    \label{tab:real_world}
\end{table}

\section*{Conclusion and Perspectives}

In this work, we have established a unified theoretical framework characterizing the conditional coverage gap in conformal prediction. We have shown that achieving exact conditional validity is equivalent to constructing nonconformity scores that are distributionally invariant with respect to the features, and interpreted the well-known finite-sample impossibility results through this lens. In addition, we introduced the PIT-CP procedure, providing flexible post-processing algorithms aiming to restore approximate conditional validity with explicit oracle bounds and convergence rates.

A limitation of this work is that our empirical evaluation depends on the particular baseline implementations we selected. Stronger baselines may have been possible with alternative implementations, more extensive hyperparameter tuning, or explicit incorporation of structural prior knowledge.

Future research directions include extending this post-processing perspective to handle atomic distributions and exploring the link with multivariate scores used in optimal transport conformal prediction, as well as potential applications to fairness and causality.

\section*{Acknowledgements}

The author wishes to thank Pierre Humbert for numerous valuable comments and remarks regarding this work.

\bibliographystyle{abbrvnat}
\bibliography{sources}

\newpage

\appendix

\section{Proofs}

\subsection{Proof of Theorem~\ref{thm:cond_coverage}} \label{pf:cond_coverage}

\begin{proof}
We first show the implication $\mathrm{(2)} \implies \mathrm{(1)}$. By the definition of the split conformal prediction region, the conditional coverage probability simplifies to the probability of the test score falling below the threshold, and thus for all $n \geq 1$, $\alpha \in (0, 1)$, and $P_X$-almost all $x \in \mathcal{X}$
\[
\mathbb{P}\left( Y_{n + 1} \in \widehat{C}_{1 - \alpha}(x) \mid X_{n + 1} = x \right) = \mathbb{P}(S_{n + 1} \leq \hat{q}_{1 - \alpha} \mid X_{n + 1} = x).
\]

Under the distributional invariance assumption (Definition~\ref{def:pivotal}), the feature $X_{n + 1}$ is statistically independent of its corresponding score $S_{n + 1}$. Furthermore, by Assumption~\ref{ass:iid}, since $\hat{q}_{1 - \alpha}$ is $\mathcal{D}_n$-measurable and $\mathcal{D}_n$ is independent of both $X_{n + 1}$ and $S_{n + 1}$, we get
\[
\mathbb{P}(S_{n + 1} \leq \hat{q}_{1 - \alpha} \mid X_{n + 1} = x) = \mathbb{P}(S_{n + 1} \leq \hat{q}_{1 - \alpha}).
\]

We now prove $\mathrm{(1)} \implies \mathrm{(2)}$. For each $n \geq 1$, choose $\alpha_n = \frac{1}{n + 1}$ so that $k = n$. Therefore, $\hat{q}_{1 - \alpha_n} = \max_{1 \leq i \leq n} S_i \eqqcolon S_{(n)}$. For this choice of $\alpha_n$, $\mathrm{(1)}$ gives
\[
\forall n \geq 1, \, \mathbb{P}\left( S_{n + 1} \leq S_{(n)} \mid X_{n + 1} \right) \overset{\mathrm{a.s.}}{=} \mathbb{P}\left( S_{n + 1} \leq S_{(n)} \right).
\]

Letting the nondecreasing function $Q_S: u \mapsto \inf \left\{ t \in \mathbb{R} : \mathbb{P}(S \leq t) \geq u \right\}$ (\textit{i.e.}, the quantile function of $S$), $U_1, \dots, U_n \overset{\mathrm{i.i.d.}}{\sim} \mathrm{Unif}(0, 1)$ as well as $U_{(n)} \coloneqq \max_{1 \leq i \leq n} U_i$, we have
\[
S_{(n)} \overset{d}{=} Q_S(U_{(n)}).
\]

Therefore, since $U_{(n)}$ has density $nu^{n-1}\mathbb{I}\{ u \in (0, 1) \}$ with respect to the Lebesgue measure, we obtain
\[
\forall n \geq 1, \, \int_0^1 \left[ \mathbb{P}\left( S \leq Q_S(u) \mid X \right) - \mathbb{P}\left( S \leq Q_S(u) \right) \right] n u^{n-1} \, du = 0 \quad P_X \text{-a.s.}
\]

Since this holds for countably many $n$, it also holds $P_X$-almost surely simultaneously for all $n \geq 1$. Also, polynomials are dense in $C\left( [0, 1] \right)$ by Stone--Weierstrass' theorem, which itself is dense in $L^2(0, 1)$, and so it follows that for $P_X \otimes \mathrm{Leb}$-almost all $(x, u) \in \mathcal{X} \times (0, 1)$
\begin{equation} \label{eq:quantile}
\mathbb{P}\left( S \leq Q_S(u) \mid X = x \right) = \mathbb{P}\left( S \leq Q_S(u) \right).
\end{equation}

It remains to pass from equality at marginal quantiles to equality at arbitrary thresholds. Fix $x \in \mathcal{X}$ such that Equation~\eqref{eq:quantile} holds for almost all $u \in (0, 1)$, and let $t \in \mathbb{R}$. If $u < F_S(t)$, then by the definition of the generalized inverse, $Q_S(u) \leq t$, and thus
\[
F_{S \mid X = x}(t) \geq F_{S \mid X = x}\left( Q_S(u) \right) = F_S\left( Q_S(u) \right)
\geq u.
\]
Letting $u \to F_S(t)$ yields
\[
F_{S \mid X = x}(t) \geq F_S(t).
\]

Moreover, we have
\[
\mathbb{E}\left[ F_{S \mid X}(t) - F_S(t) \right] = 0,
\]
hence $F_{S \mid X = x}(t) = F_S(t)$ for $P_X$-almost all $x \in \mathcal{X}$ and every fixed $t$. By considering the countable (yet dense) set of rational thresholds $t \in \mathbb{Q}$ so that the above equality holds simultaneously for all such $t$, and using the right-continuity of distribution functions, we conclude that
\[
F_{S\mid X = x} = F_S
\]
for $P_X$-almost all $x \in \mathcal{X}$, that is, $S$ and $X$ are independent.
\end{proof}

\subsection{Proof of Lemma~\ref{lem:ks_bound}} \label{pf:ks_bound}

\begin{proof}
Fix $\alpha \in (0, 1)$. By the definition of the split conformal prediction region and Assumption~\ref{ass:iid}, the empirical threshold $\hat{q}_{1 - \alpha}$ is a random variable independent of both $X_{n + 1}$ and $S_{n + 1}$. Following the logic in Theorem~\ref{thm:cond_coverage}, we can express both the conditional and marginal coverage probabilities as expectations over the distribution of the threshold. Specifically, we have
\begin{align*}
&\mathbb{P}\left( Y_{n + 1} \in \widehat{C}_{1 - \alpha}(x) \mid X_{n + 1} = x \right) = \mathbb{P}(S_{n + 1} \leq \hat{q}_{1 - \alpha} \mid X_{n + 1} = x) = \mathbb{E}\left[ F_{S \mid X = x}(\hat{q}_{1 - \alpha}) \right], \\
&\mathbb{P}\left( Y_{n + 1} \in \widehat{C}_{1 - \alpha}(X_{n + 1}) \right) = \mathbb{P}(S_{n + 1} \leq \hat{q}_{1 - \alpha}) = \mathbb{E}\left[ F_S(\hat{q}_{1 - \alpha}) \right],
\end{align*}
where the expectation is taken over the distribution of $\mathcal{D}_n$, which in turn governs that of $\hat{q}_{1 - \alpha}$. 

Substituting these expressions into the definition of $\Delta(x)$, we obtain
\begin{align*}
\Delta(x) &= \sup_{\alpha \in (0, 1)} \left\vert \mathbb{E}\left[ F_{S \mid X = x}(\hat{q}_{1 - \alpha}) - F_S(\hat{q}_{1 - \alpha}) \right] \right\vert \\
&\leq \mathbb{E}\left[ \sup_{\alpha \in (0, 1)} \left\vert F_{S \mid X = x}(\hat{q}_{1 - \alpha}) - F_S(\hat{q}_{1 - \alpha}) \right\vert \right] \\
&\leq \mathbb{E}\left[ \sup_{q \in \mathbb{R}} \left\vert F_{S \mid X = x}(q) - F_S(q) \right\vert \right] \\
&= d_{KS}(F_{S \mid X = x}, F_S).
\end{align*}
\end{proof}

\subsection{Proof of Proposition~\ref{prop:pit_invariance}} \label{pf:pit_invariance}

\begin{proof}
We evaluate the transformed score on a true random observation $Y \sim P_{Y \mid X = x}$. Let $S_x = s(x, Y)$ be the random base score, and let $F_{S_x}$ be its true cumulative distribution function. 

By Definition~\ref{def:pit_correction}, the transformed score is exactly $\tilde{s}(x, Y) = F_{S_x}(S_x)$. By the Probability Integral Transform (Theorem~1 of \cite{angus1994probability}), under Assumption~\ref{ass:cont_cond}, applying a continuous PIT to its own underlying random variable yields a standard uniform distribution. 

Thus, conditionally on $X = x$, the score $\widetilde{S} \mid X = x \sim \mathrm{Unif}(0, 1)$ for all $x \in \mathcal{X}$. Because the conditional distribution of $\widetilde{S}$ is identical for every realization of $X$, the random variable $\widetilde{S}$ is independent of $X$, satisfying Definition~\ref{def:pivotal}.
\end{proof}

\subsection{Proof of Lemma~\ref{lem:plugin_upper}} \label{pf:plugin_upper}

\begin{proof}
Fix $x \in \mathcal{X}$ such that $F_{\widehat{S} \mid X = x}$ is continuous, following Assumption~\ref{ass:plugin_cont_cond}. For the first inequality~\eqref{eq:ks_bound}, by the triangle inequality, letting $U$ be the CDF of a $\mathrm{Unif}(0, 1)$ random variable, we use Lemma~\ref{lem:ks_bound} and
\begin{equation} \label{eq:triangle}
d_{KS}(F_{\widehat{S} \mid X = x}, F_{\widehat{S}}) \leq d_{KS}(F_{\widehat{S} \mid X = x}, U) + d_{KS}(U, F_{\widehat{S}}).
\end{equation}

For the first term of Equation~\eqref{eq:triangle}, since $\widehat{S}$ takes values in $(0, 1)$, one gets
\begin{align*}
d_{KS}(F_{\widehat{S} \mid X = x}, U) &= \sup_{t \in \mathbb{R}}\left\vert F_{\widehat{S} \mid X = x}(t) - U(t) \right\vert \\
&= \sup_{t \in (0, 1)}\left\vert F_{\widehat{S} \mid X = x}(t) - U(t) \right\vert \\
&= \sup_{t \in (0, 1)}\left\vert \mathbb{P}\left( \widehat{F}_{S \mid X = x}(S) \leq t \mid X = x \right) - t \right\vert.
\end{align*}

Because $\widehat{F}_{S \mid X = x}$ is continuous, it satisfies the intermediate value property and its image covers the whole interval $(0, 1)$. Substituting $t = \widehat{F}_{S \mid X = x}(y)$ for $y \in \mathbb{R}$, by a change of variables, we get
\[
d_{KS}(F_{\widehat{S} \mid X = x}, U) = \sup_{y \in \mathbb{R}}\left\vert \mathbb{P}\left( \widehat{F}_{S \mid X = x}(S) \leq \widehat{F}_{S \mid X = x}(y) \mid X = x \right) - \widehat{F}_{S \mid X = x}(y) \right\vert.
\]

Now, for any fixed $y \in \mathbb{R}$ let us define
\[
y^* = \sup\left\{ u \in \mathbb{R} : \widehat{F}_{S \mid X = x}(u) = \widehat{F}_{S \mid X = x}(y) \right\}.
\]

Then, since $\widehat{F}_{S \mid X = x}$ is nondecreasing, the first term evaluates to
\[
\mathbb{P}\left( \widehat{F}_{S \mid X = x}(S) \leq \widehat{F}_{S \mid X = x}(y) \mid X = x \right) = \mathbb{P}( S \leq y^* \mid X = x) = F_{S \mid X = x}(y^*),
\]
and by definition we also have
\[
\widehat{F}_{S \mid X = x}(y) = \widehat{F}_{S \mid X = x}(y^*).
\]

Thus, we recover
\begin{equation} \label{eq:ks_uniform}
d_{KS}(F_{\widehat{S} \mid X = x}, U) \leq \sup_{y \in \mathbb{R}}\left\vert F_{S \mid X = x}(y) - \widehat{F}_{S \mid X = x}(y) \right\vert = d_{KS}\left( F_{S \mid X = x}, \widehat{F}_{S \mid X = x} \right).
\end{equation}

For the second term of Equation~\eqref{eq:triangle}, simply notice that for any $t \in \mathbb{R}$, by the law of total expectation, we have
\[
F_{\widehat{S}}(t) = \mathbb{E}\left[ F_{\widehat{S} \mid X}(t) \right],
\]
thus, leveraging our previous Equation~\eqref{eq:ks_uniform}
\begin{align*}
d_{KS}(U, F_{\widehat{S}}) &= \sup_{t \in (0, 1)}\left\vert \mathbb{E}\left[ t - F_{\widehat{S} \mid X}(t) \right] \right\vert \\
&\leq \mathbb{E}\left[ \sup_{t \in (0, 1)}\left\vert t - F_{\widehat{S} \mid X}(t) \right\vert \right] \\
&\leq \mathbb{E}\left[ d_{KS}\left( F_{S \mid X}, \widehat{F}_{S \mid X} \right) \right].
\end{align*}

\bigbreak

To obtain the inequality~\eqref{eq:kl_bound}, because the Kolmogorov--Smirnov distance is upper-bounded by the total variation distance, we sequentially apply Pinsker's inequality, Jensen's inequality, as well as the fact that $(a + b)^2 \leq 2(a^2 + b^2)$ for any real numbers $a, b$, and thus
\begin{align*}
d_{KS}(F_{\widehat{S} \mid X = x}, F_{\widehat{S}}) &\leq d_{TV}\left( P_{S \mid X = x}, \widehat{P}_{S \mid X = x} \right) + \mathbb{E}\left[ d_{TV}\left( P_{S \mid X}, \widehat{P}_{S \mid X} \right) \right] \\
&\leq \sqrt{\frac{1}{2} D_{KL}\left( P_{S \mid X = x} \parallel \widehat{P}_{S \mid X = x} \right)} + \mathbb{E}\left[ \sqrt{\frac{1}{2} D_{KL}\left( P_{S \mid X} \parallel \widehat{P}_{S \mid X} \right)} \right] \\
&\leq \sqrt{\frac{1}{2} D_{KL}\left( P_{S \mid X = x} \parallel \widehat{P}_{S \mid X = x} \right)} + \sqrt{\frac{1}{2} \mathbb{E}\left[ D_{KL}\left( P_{S \mid X} \parallel \widehat{P}_{S \mid X} \right) \right]} \\
&\leq \sqrt{2\left( \frac{1}{2} D_{KL}\left( P_{S \mid X = x} \parallel \widehat{P}_{S \mid X = x} \right) + \frac{1}{2} \mathbb{E}\left[ D_{KL}\left( P_{S \mid X} \parallel \widehat{P}_{S \mid X} \right) \right] \right)}.
\end{align*}
\end{proof}

\subsection{Proof of Theorem~\ref{thm:plugin_bounds}} \label{pf:plugin_bounds}

\begin{proof}
To prove the bounds, simply notice that since Assumption~\ref{ass:iid} holds, then one can bound, for all $x \in \mathcal{X}$
\[
\widehat{\Delta}(x) \leq d_{KS}(F_{\widehat{S} \mid X = x}, F_{\widehat{S}}).
\]

Therefore, applying Equations~\eqref{eq:ks_bound} or \eqref{eq:kl_bound} from Lemma~\ref{lem:plugin_upper} under the same assumptions yields the stated results.
\end{proof}

\subsection{Proof of Corollary~\ref{cor:exp_prob_bounds}} \label{pf:exp_prob_bounds}

\begin{proof}
Because the coverage gap is nonnegative and $\delta > 0$, applying Markov's inequality yields
\[
\mathbb{P}\left( \widehat{\Delta}(X) \geq \delta \right) \leq \frac{\mathbb{E}\left[ \widehat{\Delta}(X) \right]}{\delta}, \qquad \mathbb{P}\left( \widehat{\Delta}(X) \geq \delta \right) \leq \frac{\mathbb{E}\left[ \widehat{\Delta}(X)^2 \right]}{\delta^2}.
\]

Invoking the bounds of Theorem~\ref{thm:plugin_bounds} under the same assumptions directly yields the stated results, both in expected value and high probability.
\end{proof}

\subsection{Proof of Lemma~\ref{lem:oracle_inclusion}} \label{pf:oracle_inclusion}

\begin{proof}
Fix $x \in \mathcal{X}$. By the triangle inequality, for any realization of the threshold $\hat{q}_{1 - \alpha} < +\infty$
\[
\left\vert F_{\widehat{S} \mid X = x}(\hat{q}_{1 - \alpha}) - (1 - \alpha) \right\vert \leq \left\vert F_{\widehat{S} \mid X = x}(\hat{q}_{1 - \alpha}) - F_{\widehat{S}}(\hat{q}_{1 - \alpha}) \right\vert + \left\vert \widehat{F}_n(\hat{q}_{1 - \alpha}) - (1 - \alpha) \right\vert + \left\vert \widehat{F}_n(\hat{q}_{1 - \alpha}) - F_{\widehat{S}}(\hat{q}_{1 - \alpha}) \right\vert.
\]

The first term is bounded by the Kolmogorov--Smirnov distance $d_{KS}(F_{\widehat{S} \mid X = x}, F_{\widehat{S}})$. For the second term, let $\widehat{F}_n$ denote the empirical cumulative distribution function of the $n$ calibration scores. By the definition of the split conformal threshold, and since ties in $\{ \widehat{S}_i \}_{i=1}^n$ occur with zero probability, almost surely we have
\begin{align*}
\left\vert \widehat{F}_n(\hat{q}_{1 - \alpha}) - (1 - \alpha) \right\vert &= \left\vert \frac{1}{n} \left\lceil (n + 1)(1 - \alpha) \right\rceil - (1 - \alpha) \right\vert \\
&\leq \left\vert \frac{(n + 1)(1 - \alpha)}{n} - (1 - \alpha) \right\vert + \frac{1}{n} \\
&= \frac{1 - \alpha}{n} + \frac{1}{n} \leq \frac{2}{n}.
\end{align*}

For the third term, by the Dvoretzky--Kiefer--Wolfowitz inequality \citep{massart1990tight}, with probability at least $1 - \delta$ over $\mathcal{D}_n$, the empirical cumulative distribution function satisfies, regardless of $x \in \mathcal{X}$
\[
\sup_{t \in \mathbb{R}} \left\vert \widehat{F}_n(t) - F_{\widehat{S}}(t) \right\vert \leq \sqrt{\frac{\log(2/\delta)}{2n}}.
\]

Combining these guarantees gives a high probability bound on the realized conditional coverage
\[
1 - \alpha - \widehat{L}(x, n, \delta) \leq F_{\widehat{S} \mid X = x}(\hat{q}_{1 - \alpha}) \leq 1 - \alpha + \widehat{L}(x, n, \delta),
\]
which clearly also holds for $\hat{q}_{1 - \alpha} = +\infty$ because it implies that $\alpha < \frac{1}{n + 1}$, and thus $\widehat{L}(x, n, \delta) \geq \frac{2}{n} > \alpha$.

Under Assumption~\ref{ass:cont_cond}, the oracle family $\{ C^*_{1 - \beta}(x) \}_{\beta \in (0, 1)}$ is the family of sublevel sets of $s(x, \cdot)$ indexed by their conditional coverage level. Since $\widehat{C}_{1 - \alpha}(x)$ is also a sublevel set of $s(x, \cdot)$, it can be written exactly as $\{ y \in \mathcal{Y} \mid s(x, y) \leq \hat{t}_{1 - \alpha} \}$ for some effective base threshold $\hat{t}_{1 - \alpha} \in \mathbb{R} \cup \{ +\infty \}$. Its conditional coverage $F_{S \mid X = x}(\hat{t}_{1 - \alpha}) = F_{\widehat{S} \mid X = x}(\hat{q}_{1 - \alpha})$ lying between
\[
1 - \alpha - \widehat{L}(x, n, \delta) \qquad \text{and} \qquad 1 - \alpha + \widehat{L}(x, n, \delta)
\]
implies that, under Assumption~\ref{ass:increasing}, the strictly increasing conditional cumulative distribution function admits a well-defined inverse. Applying this inverse to the coverage bounds restricts the effective threshold to
\[
q^*_{1 - \alpha - \widehat{L}(x, n, \delta)}(x) \leq \hat{t}_{1 - \alpha} \leq q^*_{1 - \alpha + \widehat{L}(x, n, \delta)}(x),
\]
which directly establishes the geometric set inclusions
\[
C^*_{1 - \alpha - \widehat{L}(x, n, \delta)}(x) \subseteq \widehat{C}_{1 - \alpha}(x) \subseteq C^*_{1 - \alpha + \widehat{L}(x, n, \delta)}(x).
\]

Since the high probability argument is made irrespective of $x \in \mathcal{X}$, this result holds with high probability uniformly in $x$.
\end{proof}

\subsection{Proof of Theorem~\ref{thm:sym_diff}} \label{pf:sym_diff}

\begin{proof}
Fix $\delta \in (0, 1)$. By Lemma~\ref{lem:oracle_inclusion}, with probability at least $1 - \delta$ over $\mathcal{D}_n$, for all $x \in \mathcal{X}$
\[
C^*_{1 - \alpha - \widehat{L}(x, n, \delta)}(x) \subseteq \widehat{C}_{1 - \alpha}(x) \subseteq C^*_{1 - \alpha + \widehat{L}(x, n, \delta)}(x).
\]

Since the oracle family $\{ C^*_{1 - \beta}(x) \}_{\beta \in (0, 1)}$ consists of nested sublevel sets of $s(x, \cdot)$, we obtain
\[
\widehat{C}_{1 - \alpha}(x) \symdif C^*_{1 - \alpha}(x) \subseteq C^*_{1 - \alpha + \widehat{L}(x, n, \delta)}(x) \setminus C^*_{1 - \alpha - \widehat{L}(x, n, \delta)}(x).
\]

By definition of the oracle conditional quantiles and Assumption~\ref{ass:cont_cond}, the conditional probability of the right-hand side is bounded by the level gap, yielding
\[
\mathbb{P}\left( Y \in C^*_{1 - \alpha + \widehat{L}(x, n, \delta)}(x) \setminus C^*_{1 - \alpha - \widehat{L}(x, n, \delta)}(x) \mid X = x \right) \leq 2\widehat{L}(x, n, \delta).
\]

Therefore, combining the results up to now, with probability at least $1 - \delta$ over $\mathcal{D}_n$
\begin{equation}
\text{for } P_X\text{-almost all } x \in \mathcal{X}, \, \mathbb{P}\left( Y \in \widehat{C}_{1 - \alpha}(x) \symdif C^*_{1 - \alpha}(x) \mid \mathcal{D}_n, X = x \right) \leq 2\widehat{L}(x, n, \delta). \label{eq:upper}
\end{equation}

We then conclude by taking the expectation over $X$ and $\mathcal{D}_n$ on both sides. Indeed, noting $\Omega_\delta$ the event on which \eqref{eq:upper} is satisfied, we have
\begin{align*}
\mathbb{P}\left( Y \in \widehat{C}_{1 - \alpha}(X) \symdif C^*_{1 - \alpha}(X) \right) &\leq 2\mathbb{E}\left[ \mathbb{I}\{ \Omega_\delta \} \widehat{L}(X, n, \delta) \right] + \mathbb{P}(\Omega_\delta^c) \\
&\leq 2\mathbb{E}\left[ \widehat{L}(X, n, \delta) \right] + \delta \\
&\leq \underbrace{2\mathbb{E}\left[ d_{KS}\left( F_{\widehat{S} \mid X}, F_{\widehat{S}} \right) \right]}_{\text{estimation error}} + \underbrace{\sqrt{\frac{2\log(2/\delta)}{n}} + \frac{4}{n} + \delta}_{\text{statistical error}}.
\end{align*}

Taking the \textit{infimum} over $\delta$ essentially concludes the proof, and in particular, choosing $\delta \asymp n^{-1/2}$ achieves the stated rate. Thanks to Assumption~\ref{ass:plugin_cont_cond}, the term $\mathbb{E}\left[ d_{KS}\left( F_{\widehat{S} \mid X}, F_{\widehat{S}} \right) \right]$ can then be further bounded through Lemma~\ref{lem:plugin_upper}, and using $2\sqrt{2} \leq 3$ together with Jensen's inequality.
\end{proof}

\subsection{Proof of Theorem~\ref{thm:minimax_high_prob_bound}} \label{pf:minimax_high_prob_bound}

\begin{proof}
Theorem~3.5 of \cite{li2022minimax}, whose conditions are met in our setting, ensures the existence of an estimator $\hat{p}$ depending on the $N$ i.i.d. samples $\mathcal{D}_N'$ satisfying
\[
\mathbb{E}\left[ \iint \left\vert \hat{p}(s \mid x) - p(s \mid x) \right\vert p_X(x) \, dx \, ds \right] \lesssim N^{-\frac{1}{1/\beta + p/\gamma + 2}},
\]
where the expectation is taken over the distribution of $\mathcal{D}_N'$. 

Since the Kolmogorov--Smirnov distance is bounded by total variation, and together with Corollary~\ref{cor:exp_prob_bounds}, for every $\delta>0$, if $\widehat{F}_{S \mid X}$ is constructed using the estimator $\hat{p}$ through $\widehat{F}_{S \mid X = x}(s) \coloneqq \int_{t \leq s} \hat{p}(t \mid x) \, dt$ (which is continuous, and thus satisfies Assumption~\ref{ass:plugin_cont_cond}), for the first claim $(\mathrm{i})$, we have
\[
\mathbb{P}\left( \widehat{\Delta}(X) \geq \delta \right) \leq \frac{2\mathbb{E}\left[ d_{KS}\left( F_{S \mid X}, \widehat{F}_{S \mid X} \right) \right]}{\delta} \leq \frac{2\mathbb{E}_{\mathcal{D}_N'}\left[ \mathbb{E}_X\left[ d_{TV}\left( P_{S \mid X}, \widehat{P}_{S \mid X} \right) \mid \mathcal{D}_N'\right] \right]}{\delta} \lesssim \frac{N^{-\frac{1}{1/\beta+p/\gamma+2}}}{\delta},
\]
hence
\[
\widehat{\Delta}(X) = O_P\left( N^{-\frac{1}{1/\beta+p/\gamma+2}} \right).
\]

The second claim for $(\mathrm{ii})$ follows similarly from Theorem~\ref{thm:sym_diff} under Assumption~\ref{ass:increasing}.
\end{proof}

\subsection{Proof of Lemma~\ref{lem:circ_invariance}} \label{pf:circ_invariance}

\begin{proof}
Let $k = \lceil (n + 1) (1 - \alpha) \rceil$. First, if $k = n + 1$, then the stated result clearly holds regardless of the nonconformity score. Otherwise, let $S_i \coloneqq s(X_i, Y_i)$ for $i \in \llbracket n \rrbracket$, and define
\[
\widetilde{S}_i \coloneqq g(S_i).
\]

Since $g$ is strictly increasing, it preserves the ordering of the calibration scores, and therefore the order statistics of the transformed scores satisfy
\[
\widetilde{S}_{(k)} = g(S_{(k)}),
\]
where $S_{(k)}$ denotes the $k$-th order statistic of $\{ S_i \}_{i=1}^n$ and $\widetilde{S}_{(k)}$ denotes the $k$-th order statistic of $\{ \widetilde{S}_i \}_{i=1}^n$.

The conformal prediction region obtained from the transformed score $g \circ s$ is then
\[
\widetilde{C}_{1 - \alpha}(x) = \left\{ y \in \mathcal{Y} : g(s(x,y)) \leq \tilde{S}_{(k)} \right\}.
\]

Using $\tilde{S}_{(k)} = g(S_{(k)})$ and the fact that $g$ is strictly increasing, we get
\[
g(s(x,y)) \leq g(S_{(k)}) \quad \Longleftrightarrow \quad s(x,y) \leq S_{(k)},
\]
hence
\[
\widetilde{C}_{1 - \alpha}(x) = \left\{ y \in \mathcal{Y} : s(x,y) \leq S_{(k)} \right\} = \widehat{C}_{1 - \alpha}(x).
\]
\end{proof}

\section{Details on Volume Computation}

\subsection{HPD-split} \label{subsec:hpd}

For a fixed $x \in \mathcal{X}$, the HPD-split conformal prediction region can be expressed as
\[
\widehat{C}_\mathrm{HPD}(x) = \left\{ y \in \mathcal{Y} \mid \hat{p}(y \mid x) \geq \hat{\tau}(x) \right\},
\]
where $\hat{p}(\cdot \mid x)$ is the estimated conditional density and $\hat{\tau}(x) > 0$ is the calibrated threshold.

The volume of this region is the integral of its indicator function over $\mathcal{Y}$
\[
\mathrm{Vol}\left( \widehat{C}_\mathrm{HPD}(x) \right) = \int_\mathcal{Y} \mathbb{I}\{ \hat{p}(y \mid x) \geq \hat{\tau}(x) \} \, dy.
\]

Multiplying and dividing the integrand by $\hat{p}(y \mid x)$ allows us to rewrite this integral as an expected value under the conditional density measure
\[
\mathrm{Vol}\left( \widehat{C}_\mathrm{HPD}(x) \right) = \int_\mathcal{Y} \frac{\mathbb{I}\{ \hat{p}(y \mid x) \geq \hat{\tau}(x) \}}{\hat{p}(y \mid x)} \hat{p}(y \mid x) \, dy = \mathbb{E}_{Y \sim \hat{p}(\cdot \mid x)}\left[ \frac{\mathbb{I}\{ \hat{p}(Y \mid x) \geq \hat{\tau}(x) \}}{\hat{p}(Y \mid x)} \right].
\]

This expected value can then be approximated through Monte Carlo sampling, in practice using $B = 1000$ samples.

\subsection{CONTRA} \label{subsec:contra}

For a fixed $x \in \mathcal{X}$, the CONTRA conformal region leverages a normalizing flow $f(\cdot \mid x) \colon \mathcal{Y} \to \mathcal{Z}$ and can be expressed as
\[
\widehat{C}_\mathrm{CONTRA}(x) = \left\{ y \in \mathbb{R}^d : \Vert \hat{f}(y \mid x) \Vert \leq \hat{r} \right\},
\]
where $\mathcal{Y} = \mathbb{R}^d$ for $d \geq 1$ and $\hat{r}$ is the calibrated radius in the latent space.

Its volume corresponds to
\[
\mathrm{Vol}\left( \widehat{C}_\mathrm{CONTRA}(x) \right) = \int_{\widehat{C}_\mathrm{CONTRA}(x)} 1 \, dy.
\]

By a change of variables, setting $z = \hat{f}(y \mid x)$ with Jacobian $J_{\hat{f}^{-1}}(z \mid x)$, the integration domain maps to the $d$-dimensional Euclidean ball $B_d(0, \hat{r})$
\[
\mathrm{Vol}\left( \widehat{C}_\mathrm{CONTRA}(x) \right) = \int_{B_d(0, \hat{r})} \left\vert \det J_{\hat{f}^{-1}}(z \mid x) \right\vert \, dz.
\]

Letting $V_d(\hat{r})$ denote the volume of $B_d(0, \hat{r})$, we introduce the uniform probability density over the ball to express the integral as an expected value
\[
\mathrm{Vol}\left( \widehat{C}_\mathrm{CONTRA}(x) \right) = V_d(\hat{r}) \int_{B_d(0, \hat{r})} \left\vert \det J_{\hat{f}^{-1}}(z \mid x) \right\vert \, \frac{dz}{V_d(\hat{r})} = V_d(\hat{r}) \mathbb{E}_{Z \sim \mathrm{Unif}\left( B_d(0, \hat{r}) \right)}\left[ \left\vert \det J_{\hat{f}^{-1}}(Z \mid x) \right\vert \right].
\]

This expected value can then be approximated through Monte Carlo sampling, in practice using $B = 1000$ samples.

\end{document}